\documentclass[accepted=2019-03-29]{quantumarticle}
\usepackage{hyperref}

\usepackage{amsfonts}
\usepackage{amsthm}
\usepackage{amsmath}

\newtheorem{theorem}{Theorem}

\newtheorem{lemma}[theorem]{Lemma}
\newtheorem{definition}[theorem]{Definition}
\newtheorem{proposition}[theorem]{Proposition}

\newenvironment{proofSketch}{%
  \proof}{\endproof}

\usepackage{amsmath}
\DeclareMathOperator{\sign}{sign}

\usepackage{booktabs}
\usepackage{afterpage}

\usepackage[utf8]{inputenc}
\usepackage[T1]{fontenc}
\usepackage[x11names]{xcolor}
\usepackage{array, colortbl}

\newcommand{\lam}[2]{\lambda_{#2}^{#1}}
\newcommand{\lamp}[2]{\tilde{\lambda}_{#2}^{#1}}

\usepackage{graphicx}

\newcommand{\ket}[1]{|#1\rangle}
\newcommand{\bra}[1]{\langle #1 |}

\newcommand{\zmatrix}{symmetric Z-matrix}
\newcommand{\sPi}{\mathsf{\Pi}}
\newcommand{\lowrank}{rank-1}
\newcommand{\highrank}{rank>1}
\newcommand{\basegraph}{weighted interaction graph}
\newcommand{\quotientGraph}{ \lowrank\ quotient graph}

\usepackage{wrapfig}

\begin{document}

\title{Two-local qubit Hamiltonians: when are they stoquastic?}

\author{Joel Klassen}
\affiliation{QuTech, Delft University of Technology, P.O. Box 5046, 2600 GA Delft, The Netherlands}
\author{Barbara M. Terhal}
\affiliation{QuTech, Delft University of Technology, P.O. Box 5046, 2600 GA Delft, The Netherlands}
\affiliation{Institute for Theoretical Nanoelectronics,
Forschungszentrum Juelich, D-52425 Juelich, Germany}

\maketitle

\begin{abstract}
We examine the problem of determining if a 2-local Hamiltonian is stoquastic by local basis changes. We analyze this problem for two-qubit Hamiltonians, presenting some basic tools and giving a concrete example where using unitaries beyond Clifford rotations is required in order to decide stoquasticity. We report on simple results for $n$-qubit Hamiltonians with identical 2-local terms on bipartite graphs. Our most significant result is that we give an efficient algorithm to determine whether an arbitrary $n$-qubit XYZ Heisenberg Hamiltonian is stoquastic by local basis changes. 
\end{abstract}

\section{Introduction and Motivation}

The notion of stoquastic Hamiltonians was introduced in \cite{Bravyi2007} in the context of quantum complexity theory. The definition of stoquastic Hamiltonians aims to capture Hamiltonians which do not suffer from the sign problem: the definition ensures that the ground state of the Hamiltonian has nonnegative amplitudes in some local basis while the matrix $\exp(-H/kT)$ is an entrywise nonnegative matrix for any temperature $T$ in this basis. Stoquastic Hamiltonians are quite ubiquitous. Any (mechanical) Hamiltonian defined with conjugate variables $[\hat{x}_i,\hat{p}_j]=i\hbar \delta_{ij}$ of the form
\begin{equation}
H=\frac{1}{2} {\vec{p}}^{\mathsf{T}} C^{-1} \vec{p} +U(\vec{x}).
\label{eq:mechH}
\end{equation}
with invertible (mass) matrix $C> 0$, is stoquastic, since the kinetic term is off-diagonal in the position $\vec{x}$-basis. The paradigm of circuit-QED produces such Hamiltonians by the canonical quantization of an electric circuit represented by its Lagrangian density. 

For the subclass of stoquastic Hamiltonians the problem of determining the lowest eigenvalue might be easier than for general Hamiltonians. In complexity terms the stoquastic lowest-eigenvalue problem is StoqMA-complete while it is QMA-complete for general Hamiltonians. 

This difference in complexity goes hand-in-hand with the existence of heuristic computational methods for determining the lowest eigenvalue of stoquastic Hamiltonians using Monte Carlo techniques \cite{CM:projection, Sorella2000}. 
These methods can be viewed as stochastic realizations of the power method, namely a stochastic efficient implementation of the repeated application of the nonnegative matrix $\exp(-H)$. However, there is no proof that these heuristic methods are as efficient as quantum phase estimation for estimating ground-state energies (some results are obtained in \cite{Bravyi:guiding}).

Beyond the lowest eigenvalue problem, one can also consider the complexity of evaluating the partition function $Z={\rm Tr} \left[ \exp(-H/kT)\right]$ of a stoquastic Hamiltonian. The nonnegativity of $\exp(-H/kT)$ directly ensures that one can rewrite $Z$ as the partition function of a classical model in one dimension higher, which can then be sampled using classical stochastic techniques. Again, rigorous results on this well-known path-integral quantum Monte Carlo method (see e.g. \cite{wessel:lecture}) are sparse: \cite{thesis:crosson} shows how to estimate the partition function with $kT=\Omega(1/\log n)$ for 1D stoquastic Hamiltonians using a rigorous analysis of the path-integral Monte Carlo method.
In \cite{Bravyi2016} the authors gave a fully-polynomial randomized approximating scheme for estimating the partition function of stoquastic XY models (in a transverse magnetic field). 

Stoquastic Hamiltonians play an important role in adiabatic computation \cite{AL:adiabatic}. It has been proven that there is no intrinsic quantum speed-up in stoquastic adiabatic computation which uses frustration-free stoquastic Hamiltonians: a classical polynomial-time algorithm for simulating such adiabatic computation was described in \cite{Bravyi2008}. The quantum power of the well-known quantum annealing method which uses the transverse field Ising model is still an intriguing open question, given the road-block cases that have been thrown up for the path-integral Quantum Monte Carlo method in \cite{Hastings} as well as the diffusion Monte Carlo method in \cite{bringewatt:diff}.
However, given that stoquastic adiabatic computation such as quantum annealing is amenable to heuristic Monte Carlo methods, research has also moved into the direction of finding ways of engineering non-stoquastic Hamiltonians (see e.g. \cite{kafri:nonstoq}, \cite{hormozi:anneal}, \cite{samach:talk}).

If it is true that questions pertaining to stoquastic Hamiltonians are easier to answer than those pertaining to general Hamiltonians, then it is clearly important to be to able to recognize {\em which Hamiltonians are stoquastic}.

Work aimed at recognizing where and when a sign problem does not exist, is not new, but few systematic approaches to this problem exist. Most of the work has been focused on understanding the sign problem for specific classes of Hamiltonians of physical interest, for example Fermi-Hubbard models \cite{IKS:sign}. An example of a recent exploration relating stoquasticity to time-reversal for fermionic systems is \cite{WZ:sign} (see \cite{LY:review} for a recent review). \cite{BF:MPrule} has attempted to generalize the so-called Peierls-Marshall sign rules for isotropic spin-1/2 Heisenberg models to more general XYZ Hamiltonians, falling short of formulating any algorithmic approach. More recently, \cite{Marvian2018} found the first negative algorithmic results. They showed that for 3-local qubit Hamiltonians determining whether single-qubit Clifford rotations can indeed "cure the sign problem" is NP-complete (and similarly for orthogonal rotations).

In this paper we report on our first results in this direction (see e.g. \cite{terhal:talk}): they are summarized in Section \ref{sec:summary}. Our main result is an efficient algorithm to determine whether an XYZ Hamiltonian is stoquastic. 

The definition of stoquasticity, see Definition 2, is primarily motivated by computational efficiency.
A local change of basis can be efficiently executed in that the description of the $k$-local Hamiltonian in this new basis can be easily obtained. It is however clear that the subclass of single-qubit unitary transformations are only the simplest example of a mapping or encoding of one Hamiltonian $H$ to another Hamiltonian $H_{\rm sim}$ as has been formalized in \cite{CMP:universal}. 
Such a more general mapping could allow: constant-depth circuits; ancilla qubits; perturbative gadgets and approximations such that only $\lambda(H) \approx \lambda(H_{\rm sim})$ or $Z_H \approx c Z_{H_{\rm sim}}$ for some known constant $c$. The goal of such an encoding is to map the original Hamiltonian $H$ onto a simulator Hamiltonian $H_{\rm sim}$ where $H_{\rm sim}$ is stoquastic and its low-energy dynamics effectively capture those of $H$. The upshot is that the encoded Hamiltonian $H_{\rm sim}$ is directly computationally useful. If for a class of Hamiltonians $H$, simulator Hamiltonians $H_{\rm sim}$ can be found, then the ground-state energy problem or the adiabatic computation of $H$ is amenable to quantum Monte Carlo techniques and its complexity is in StoqMA. Additionally, the definition of encoding employed in \cite{CMP:universal} ensures that the partition function of the simulator Hamiltonian can be used to estimate that of the original Hamiltonian. An example is the mapping of any $k$-local Hamiltonian onto a real $k+1$-local Hamiltonian \cite{CMP:universal}. 

Surprisingly, there are no known no-go's for the use of perturbative gadgets to map seemingly non-stoquastic Hamiltonians onto stoquastic simulator Hamiltonians. Even though mapping all Hamiltonians onto stoquastic Hamiltonians seems to be excluded from a computational complexity perspective, it is possible that there are Hamiltonians which are not stoquastic by a local basis change but for which a stoquastic simulator Hamiltonian exists. 

\subsection{Definitions}

To make contact with some of the mathematics literature, we will use some standard terminology:

\begin{definition}[Symmetric $Z$-matrix]
A matrix is a symmetric $Z$-matrix if all of its matrix elements are real, it is symmetric, and its off-diagonal elements are non-positive.
\label{def:Z}
\end{definition}

Note that for a symmetric $Z$-matrix $H$ the matrix exponential $\exp(-H/kT)$ is entrywise non-negative for all $kT\geq 0$ \footnote{It is worthwhile noting that the entrywise non-negativity of the matrix $\exp(-H/kT)$ for a given temperature $T$ does not imply that $H$ is a symmetric $Z$-matrix. It is simple to generate numerical $4 \times 4$ cases where $G=A^T A$ is a non-negative matrix by taking $A$ as a non-negative matrix, but $H=-\log G$ is not a symmetric $Z$-matrix. Thus when $H$ is a symmetric $Z$-matrix, the non-negativity of $\exp(-H/kT)$ is guaranteed for all values of $kT$, while in other cases $\exp(-H/kT)$ might be non-negative only for sufficiently small values of $kT$.}.  

We reproduce the definition of a Hamiltonian being termwise stoquastic from \cite{Bravyi2007} \footnote{\cite{Bravyi2007} points out that for 2-local Hamiltonians one can prove that termwise stoquastic is the same as stoquastic. This is however not true for $k> 2$-local Hamiltonians generally.}. We will refer to it as stoquastic for simplicity.

\begin{definition} [Stoquastic \cite{Bravyi2007}]
A k-local Hamiltonian H is stoquastic if there exists a local basis $\mathcal{B}$ and a decomposition $H = \sum_i D_i$ such that each $D_i$ is a symmetric Z-matrix in the basis  $\mathcal{B}$ and acts non-trivially on at most k qudits.
\label{def:stoq}
\end{definition}

Throughout the text the terms local basis, local change of basis, local unitary, or local Clifford, are all 1-local in the sense that they preserve the 1-local tensor product structure of the Hilbert space.

Given the discussion in the introduction and following \cite{Marvian2018}, we also introduce the following definition (loosely stated):

\begin{definition}
A family of Hamiltonians $\{H\}$ on $n$ qudits is said to be \textbf{computationally stoquastic} if there is a polynomial-time algorithm for (approximately) mapping $\{H\}$ onto a family of stoquastic simulator Hamiltonians $\{H_{\rm sim}\}$. The description of $H_{\rm sim}$ in terms of its entries should be efficiently given. 
\end{definition}

Another aspect of stoquasticity is the complexity of determining whether $H$ is (computationally) stoquastic. The question is only clearly formulated when we restrict ourselves to specific mappings. For example, one can ask: is there an efficient algorithm for determining whether a 2-local Hamiltonian on $n$ qubits is stoquastic in a basis obtained by single-qubit Clifford rotations. In this light, it was shown in \cite{Marvian2018} that for 3-local qubit Hamiltonians determining whether single-qubit Clifford rotations can indeed "cure the sign problem" is NP-complete (similarly for orthogonal rotations). For two-local Hamiltonians the question how hard it is to decide whether a Hamiltonian is stoquastic by local basis changes is entirely open. 
The next section reviews our main results on this matter.

\subsection{Summary of Results}
\label{sec:summary}

In Section \ref{twoQubit} we consider, as a warm-up, the two-qubit Hamiltonian. We introduce some concepts which we will make use of in the multi-qubit setting, and we analyze under what conditions a two-qubit Hamiltonian can be transformed into a $Z$-matrix by a local basis change. This turns out to be surprisingly complex. In Theorem \ref{theo:real} we show that one can determine whether a 2-qubit Hamiltonian is real under local basis changes by inspecting a subset of non-local invariants. This thus gives a simple necessary condition to check if one wants to decide whether a two-qubit Hamiltonian is stoquastic.
 In Section \ref{cliffordInsufficient} we establish that basis changes beyond single-qubit Clifford rotations are, perhaps unsurprisingly, strictly stronger than Clifford rotations. For the general two-qubit case we show that a two-qubit Hamiltonian $H$ is stoquastic iff one can find a solution to a set of degree-4 polynomials in two variables (the polynomials are quadratic in each variable), with details in Appendix \ref{inequalityAppendix}.

In Section \ref{polyAlgoXYZ} we move to the problem of determining if an $n$-qubit 2-local Hamiltonian is stoquastic. The challenge in this case is twofold. Recalling Definition \ref{def:stoq} one must find an appropriate {\em decomposition} into 2-qubit terms $D_i$. Additionally one must find a consistent assignment of rotations to bring the two-qubit terms into $Z$-matrix form. The interplay of these two tasks seems to make the problem quite difficult. As such, we do not try to solve the most general case. Instead we consider tractable cases of this problem. The first case we consider is deciding if a multi-qubit 2-local Hamiltonian admits a decomposition into 2-qubit terms each of which are $Z$-matrices. This case is tractable because we are ignoring any local unitary transformations. A more general case of this problem is also discussed in \cite{Marvian2018}, but here we specifically discuss the two-local case. Next we consider a family of natural Hamiltonians, which have identical two-local terms on a bipartite interaction graph. We prove that in this case the problem of finding a set of one-local unitary rotations in concert with finding an appropriate decomposition is tractable, due to the special structure of the Hamiltonian.

The final case we consider is a subclass of Hamiltonians, called XYZ Heisenberg Hamiltonians, with an arbitrary interaction graph. Here we present an efficient algorithm to decide if a given Hamiltonian  of this type is stoquastic, and to construct the unitary rotation which transforms the Hamiltonian into a $Z$-matrix. This constitutes the most significant result of this paper. For such Hamiltonians, finding a decomposition is trivial, because all 1-local terms are zero. This leaves us only needing to consider how to apply one-local unitary rotations. The Hamiltonian of the XYZ Heisenberg model for $n$ qubits is $H=\sum_{u,v} H_{uv}$ for $u,v=1,\ldots,n$ where
\begin{equation}
 H_{uv} = a_{XX}^{uv} X_u X_v+ a_{YY}^{uv} Y_u Y_v+a_{ZZ}^{uv} Z_u Z_v.
\end{equation}
Here the coefficients $a^{uv}_{PP}$, which can be different for each $uv$, are specified with some $k$ bits.  The theorem that we prove in Section \ref{sec:proof-XYZ} is:

\begin{theorem}
There is an efficient algorithm that runs in time $O(n^3)$ to decide whether an $n$-qubit ${\rm XYZ}$ Heisenberg model Hamiltonian is stoquastic. The algorithm finds the local basis change such that the resulting $H$ is a symmetric $Z$-matrix or decides that it does not exist.
\label{thm:eff-stoq}
\end{theorem}

This result relies critically on Lemmas \ref{cliffordsSuffice} and \ref{cliffDiag} (proven in Appendix \ref{permsandrefs}) which show that when an XYZ Heisenberg model is stoquastic, it can always be transformed into a $Z$-matrix by single-qubit Clifford rotations. Clifford rotations are significant here because single-qubit Clifford rotations are a finite group, and thus dramatically simplify the analysis.

Going beyond the XYZ Heisenberg Hamiltonian, one can ask about the complexity of deciding stoquasticity by local Clifford rotations. Even though this is not the most general class of local basis changes, see for example Equation \eqref{counter} in Section \ref{twoQubit}, our algorithm for the XYZ Hamiltonian suggests progress could most readily be made in this direction. We are able to prove that determining whether a Hamiltonian is real by local Clifford rotations is NP-complete:

\begin{theorem} \label{thrm:NPcomplete}
Deciding if a 2-local $n$-qubit Hamiltonian is real under single-qubit Clifford operations is NP-complete: there exists a subclass of 2-local $n$-qubit Hamiltonians for which deciding if they are real under single-qubit Clifford operations is equivalent to solving the restricted exact covering by 3-sets problem (RXC3), which is known to be NP-complete.
\end{theorem}

This is significant because making a Hamiltonian real is a necessary condition for transforming it into a $Z$-matrix. Note that this result does not imply that deciding if a Hamiltonian can be transformed into a $Z$-matrix by single-qubit Clifford rotations or even general local basis changes for 2-local Hamiltonians is NP-complete. 
We give the proof of Theorem \ref{thrm:NPcomplete} in Appendix \ref{realnessIsNPComplete} since it is not central to the paper.

\section{Warm-Up: Two Qubit Hamiltonians} \label{twoQubit}
In this section we establish some basic facts about two-qubit Hamiltonians which we will make use of throughout the paper. We begin by establishing under what conditions a two-qubit Hamiltonian is a Z-matrix. We then present a convenient representation of two-qubit Hamiltonians which we will make heavy use of throughout the paper. We briefly discuss the nature of Clifford transformations in this representation, and give an example where Clifford transformations are not sufficient to decide if a Hamiltonian is stoquastic. Next we present a set of invariants which completely indicate if a Hamiltonian can be made real under local unitary rotations, a necessary condition for a Hamiltonian to be stoquastic. Finally we explain how the problem of deciding if a two-qubit Hamiltonian is stoquastic can be reduced to deciding if a set of degree-4 polynomial inequalities admit a solution, which suggests that even in the two-qubit case deciding if a Hamiltonian is stoquastic is a non-trivial task. 

Consider a general two qubit Hamiltonian
\begin{equation}
H=\sum_{ij=I, X, Y, Z} a_{ij}\, \sigma_i \otimes \sigma_j   \;, \; a_{ij} \in \mathbb{R},
\end{equation}
where we may take $a_{II}=0$ without loss of generality.
It is simple to show that 
\begin{proposition}\label{stoquasticCondition}
$H$ is a \zmatrix{} if and only if $a_{IY}=a_{YI}=a_{XY}=a_{YX}=a_{ZY}=a_{YZ}=0$ (the matrix is real) and $a_{XX} \leq - \vert a_{YY} \vert$ and $a_{IX} \leq -\vert a_{ZX} \vert$, $a_{XI} \leq -\vert a_{XZ} \vert$ (the matrix has non-positive off-diagonal elements).
\end{proposition}

\begin{proofSketch}
Visually, the positive off-diagonal contributions of XX and YY terms need to be removed leading to $a_{XX} \leq - \vert a_{YY} \vert$. Similarly, the positive off-diagonal contributions of IX and ZX terms need to be removed, leading to $a_{IX} \leq -\vert a_{ZX} \vert$ (and similarly $a_{XI} \leq -\vert a_{XZ} \vert$).
\end{proofSketch}
The set of traceless symmetric $Z$-matrices forms a polyhedral cone $\mathcal{C}_2$, that is, a cone supported by a finite number of hyperplanes. In the case of $4 \times 4$ matrices, each matrix is determined by $9$ parameters (three of them real and 6 of them nonnegative). Viewed as a polyhedral cone, $\mathcal{C}_2$ has 12 extremal vectors  \cite{BH:ising}
\begin{alignat*}{2}
&{\bf o}_1=-X \otimes X -Y \otimes Y \;& 
&{\bf o}_2=-X \otimes X + Y \otimes Y \\
 &{\bf o}_3=-\ket{0}\bra{0} \otimes X &
&{\bf o}_4=-\ket{1}\bra{1} \otimes X \\
 &{\bf o}_5=-X \otimes \ket{0}\bra{0} & &{\bf o}_6=-X \otimes \ket{1}\bra{1}\\
&{\bf d}_{1}^{\pm}=\pm Z \otimes I && {\bf d}_2^{\pm}=\pm I \otimes Z \\
 &{\bf d}_3^{\pm}=\pm Z \otimes Z. 
\end{alignat*}  Appendix \ref{AppendixNecessaryAndSufficientConditions} discusses some of this structure and its generalization to two qudits.

\subsection{A Convenient Representation of 2-qubit Hamiltonians} \label{convRep}
Two-qubit Hamiltonians are conveniently parametrized by a $3 \times 3$ real matrix and two three-dimensional real vectors:
\begin{align}
\beta =\left( \begin{matrix}
a_{XX} & a_{XY} & a_{XZ} \\
a_{YX} & a_{YY} & a_{YZ} \\
a_{ZX} & a_{ZY} & a_{ZZ}
\end{matrix}\right) && S = \left( \begin{matrix} a_{XI} \\ a_{YI} \\ a_{ZI} \end{matrix} \right) && P = \left( \begin{matrix} a_{IX} \\ a_{IY} \\ a_{IZ} \end{matrix} \right).
\end{align}
A basis transformation $H \rightarrow (U_1 \otimes U_2) H (U_1 \otimes U_2)^{\dagger} $ corresponds to a pair of $SO(3)$ rotations $O_1$ and $O_2$ acting as
\begin{align}
\beta \rightarrow O_1 \beta O_2^T && S \rightarrow O_1S && P \rightarrow O_2P .
\end{align}
This correspondence between $SU(2)$ and $SO(3)$ is well known in the context of one qubit, and holds also in the two qubit case, 
\begin{align*}
H'&= (U_1 \otimes U_2) H (U_1 \otimes U_2)^{\dagger} \\
H'&= \sum_{i\in {XYZ}}\alpha_{iI}' \sigma_i \otimes I +\sum_{j\in {XYZ}}\alpha_{Ij}' I \otimes \sigma_j   \\&+ \sum_{ij \in {X,Y,Z}} \alpha_{ij}' \sigma_i  \otimes \sigma_j  
\end{align*}
$$ \alpha_{iI}'= \frac{1}{4} {\rm Tr}[H' \sigma_i \otimes I]=\sum_{m}  [O_1]_{im} \alpha_{mI} $$
$$ \alpha_{Ij}'= \frac{1}{4} {\rm Tr}[H' I \otimes \sigma_j]=\sum_{n}  \alpha_{In} [O_2]_{jn}$$
$$ \alpha_{ij}'= \frac{1}{4} {\rm Tr}[H' \sigma_i \otimes \sigma_j]=\sum_{mn}  [O_1]_{im} \alpha_{mn} [O_2]_{jn}$$

$$[O_x]_{ij} = \frac{1}{2}{\rm Tr}[\sigma_i U_x \sigma_j U_x^{\dagger}].$$


This implies that for 2-local qubit Hamiltonians  it suffices to consider pairs of local $SO(3)$ rotations and their effect on $\beta, P, S$. It is particularly useful to consider $\beta$ in diagonal form. Every real square matrix admits a singular value decomposition $\beta = O_L \Sigma O_R^{T}$ with $O_L,O_R \in O(3), \; \Sigma \geq 0$. If we do not enforce that $\Sigma$ be positive-semidefinite, then $O_L, O_R \in SO(3)$. Hence there always exists local unitary rotations acting on $H$ which put the $\beta$ matrix into diagonal form.

\subsection{Clifford Rotations and Beyond}
\label{cliffordInsufficient}
The single-qubit Clifford rotations play a special role as local basis changes. The action of single-qubit Clifford rotations on the Pauli-matrices form a discrete subgroup of SO(3) representing the symmetries of the cube \cite{GS:clifford}. The Clifford rotations realize any permutation of the Paulis (the group $S_3$) as well as sign-flips $\sigma_i \rightarrow \pm \sigma_i$ (with determinant 1). As we shall see later in the paper, there are some cases where it is sufficient to consider Clifford rotations alone to decide if a Hamiltonian is stoquastic. It is important to establish that this is not always the case to get some intuition about the problem. Here we will present a two-qubit Hamiltonian for which it does not suffice to consider only the single-qubit Clifford rotations. Consider the two-qubit Hamiltonian
\begin{equation}
H=a_{Z}(Z_1+Z_2)-a_{X} (X_1+X_2) +a_{XX} X_1 X_2 + Z_1 Z_2.
\label{counter}
\end{equation}
with $1 \geq a_{XX} \geq 0$ and $1>a_{Z}>0$, $a_{X} > 0$. It is easy to argue that no single-qubit Clifford basis change can make this Hamiltonian into a symmetric $Z$-matrix. Since $H$ has to be real, no permutations to $Y$-components are allowed. A sign-flip which makes the XX term negative must also flip the sign of one of the single-qubit X terms. Interchanging X and Z for both qubits leads to the same problem for the ZZ  and Z terms. Applying a sign flip on the Z of one of the qubits, followed by interchanging X and Z on that same qubit leads to the requirement  that $-a_{Z} \leq -1$ and $-a_X \leq -\vert a_{XX} \vert$. This is not satisfiable as long as $0<a_{Z} <1$.

However, we could create off-diagonal XZ and ZX terms by single-qubit non-Clifford rotations. For example, one could apply $U_1 = e^{i \pi Y/8 }e^{i \pi X/2 }$ on the first qubit and $U_2 = e^{-i \pi Y/8 }e^{i \pi X/2 }$ on the second qubit, so that $U_1 X U_1^{\dagger} = - U_2 Z U_2^{\dagger}=(Z+X)/\sqrt{2}$, $U_1 Z U_1^{\dagger} = U_2 X U_2^{\dagger}=(Z-X)/\sqrt{2}$ and $U_1 Y U_1^{\dagger} = U_2 Y U_2^{\dagger} = -Y$.  In this new basis $H$ is a symmetric $Z$-matrix when \begin{equation} 
a_{X} \geq a_{Z}, \;\; 2(a_{Z}-a_{X})^2 \geq (a_{XX}+1)^2.
\end{equation}
For $a_{XX} \leq \sqrt{2}-1$, and sufficiently large $|a_{Z}-a_{X}|$, these inequalities can certainly be satisfied.

\begin{figure}[h!]
\includegraphics[width=.95 \linewidth]{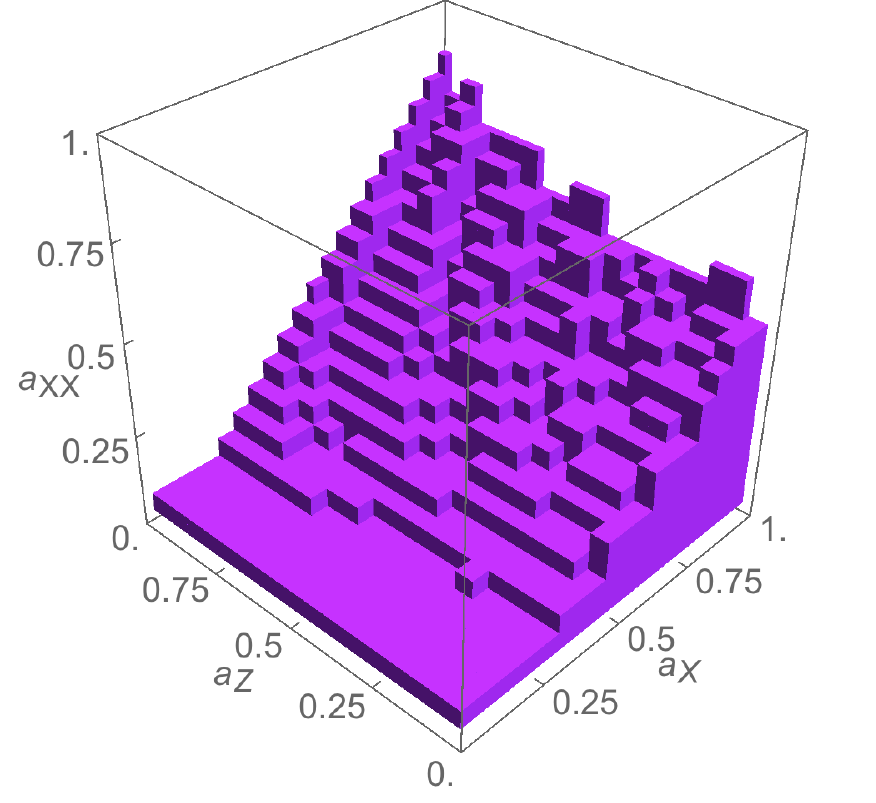}
\caption{A coarse-grained visualization of the approximate region where the Hamiltonian in Equation \eqref{counter} is stoquastic. The purple cubes correspond to points in parameter space where a solution is found to the system of equations \eqref{ineq-set1}, \eqref{ineq-set2} and \eqref{ineq-set3} when sampling a uniform lattice of points. The region obscured by the blocks is filled.}
\label{fig:3Dplot}
\end{figure}

In Figure \ref{fig:3Dplot} we show in what region of the parameter space $(a_{X}, a_{Z}, a_{XX})$ $H$ is stoquastic using the general set of inequalities given in Appendix \ref{inequalityAppendix}.

\subsection{Deciding Realness With Makhlin Invariants}\label{makhlinReal}

To characterize the set of stoquastic 2-qubit Hamiltonians one needs to consider the orbit of $\mathcal{C}_2$ under the local rotations. For a traceless Hermitian matrix acting on two qubits, there is a complete set of 18 polynomial invariants $\{I_l\}_{l=1}^{18}$ specifying the Hermitian matrix up to local unitary rotations, which we will refer to as the Makhlin invariants \cite{Makhlin2000}. Any two traceless Hamiltonians which have the same Makhlin invariants must be equivalent up to local unitary rotations, and vice versa. Since the invariants are straightforward to compute, if we are interested in some property only up to local unitary rotations, such as being stoquastic, then if we can find a way of characterizing that property in terms Makhlin invariants, it may be easier to test for that property. 
Note that the fact that there are 18 polynomial invariants can be consistent with the existence of only $d_{\rm inv}=15-6=9$ independent non-local parameters in a two-qubit traceless Hermitian matrix \cite{Linden1999,Grassl:inv}. One can see these 18 Makhlin invariants simply as a convenient set of invariants.

Stoquastic two-qubit Hamiltonians thus form a volume in the $d_{\rm inv}=9$-dimensional non-local parameter space: the question is how the invariants are constrained in this stoquastic subspace. 

We have not been able to express the stoquasticity of a two-qubit Hamiltonian in terms of inequalities involving the Makhlin or other invariants. However, we are able to capture the condition that $H$ can be made real by local rotations in terms of some of the Makhlin invariants needing to be zero, see Theorem \ref{theo:real}. This is useful because realness is a necessary condition for a Hamiltonian to be stoquastic\footnote{It is well known that any $k$-local Hamiltonian can be mapped onto a real $k+1$-local Hamiltonian, see e.g. \cite{CMP:universal}. This is not directly useful since deciding whether local basis changes exist which make a 3-local Hamiltonian have non-positive off-diagonal elements will be harder than the same problem for 2-local Hamiltonians.}.

The convenience of the Makhlin invariants is that they can be expressed in terms of inner products and triple products constructed from $\beta$, $S$ and $P$ which are invariant under $SO(3)$ rotations. Of particular interest are the triple product Makhlin invariants $I_{10}$, $I_{11}$ and $I_{15}-I_{18}$. These take the following form (using the notation that $(\vec{a},\vec{b},\vec{c})$ stands for the scalar triple product $\vec{a} \cdot (\vec{b} \times \vec{c})$):
\begin{equation}\label{invariants}
\begin{array}{cc}
\hline 
\hline
\rule{0pt}{2.5ex}   
I_{10}=&(S, \beta \beta^T S, [\beta \beta^T]^2 S)\\
I_{11}=&(P,  \beta^T \beta P, [\beta^T \beta]^2 P)\\
I_{15}=&(S, \beta \beta^T S, \beta P)\\
I_{16}=&(\beta^T S, P, \beta^T \beta P)\\
I_{17}=&(\beta^T S,  \beta^T \beta \beta^T S, P)\\
I_{18}=&(S, \beta P, \beta \beta^T \beta P) \\
 
\hline 
\hline
\end{array}
\end{equation}

One can prove that
\begin{theorem}
A two-qubit Hamiltonian $H$ is real under local unitary rotations if and only if all of the triple product invariants given in Table \ref{invariants} are equal to zero.
\label{theo:real}
\end{theorem}

\begin{proofSketch}

The proof of this Theorem is simple in one direction, namely when $H$ is real, then the triple product invariants of the corresponding $\beta, S, P$ are zero.  This is because the $Y$ coefficients in $S$ and $P$ are zero, and so they lie in a common two-dimensional subspace and the repeated multiplication by $\beta$ or $\beta^T$ keeps them in this subspace. Any triple of vectors, all lying in a two-dimensional space have zero triple product. 

In order to prove the other direction, one can observe that the triple product invariants are all equal to zero if and only if the vectors transforming under $O_L$ are co-planar, and the vectors transforming under $O_R$ are co-planar. The vectors transforming under $O_L$ are
$$\Gamma_L = \left\lbrace S, \beta \beta^T S,  [\beta \beta^T]^2 S, \beta P , \beta \beta^T \beta P \right\rbrace,$$
and the vectors transforming under $O_R$ are
$$\Gamma_R = \left\lbrace P, \beta^T\beta P , [\beta^T \beta]^2 P ,  \beta^T S  , \beta^T \beta \beta^T S \right\rbrace.$$

To prove that zero triple product invariants implies realness under local rotations is more work since $\beta$ can be degenerate and have rank less than 3, so the proof has to incorporate the various cases. We have included a full proof of Theorem \ref{theo:real} in the Appendix \ref{AppendixLocalRealnessCondition}.
\end{proofSketch}

\subsection{Reducing the Two-Qubit Stoquastic Decision Problem to Polynomial Inequalities}\label{sec:ineq}
Here we present a general analytical strategy to determine if a two qubit Hamiltonian is stoquastic, and in what basis. 

First apply the Makhlin invariants from the previous section to determine if the Hamiltonian can be made real. If not, then it is not stoquastic. 

We would now like to put our Hamiltonian into a \emph{standard form}. Recalling the discussion in section \ref{convRep} we know that it is always possible to bring $\beta$ into a diagonal form:
$$
\beta=\left(\begin{array}{ccc}
a_{XX} &0&0 \\
0& a_{YY} &0 \\
0&0& a_{ZZ}
 \end{array}\right).$$

A trivial consequence of Theorem \ref{theo:real} is that $H$ is real under local unitary rotations if and only if there exists a transformation which diagonalizes $\beta$ and puts $S$ and $P$ into the form:

\begin{align}
S=\left( \begin{array}{c}
 a_{XI} \\
 0 \\
 a_{ZI}
 \end{array} \right) ,&&
 P=\left( \begin{array}{c}
 a_{IX} \\
0 \\
 a_{IZ}
 \end{array} \right).
 \end{align}

So, given that we know our Hamiltonian can be made real under local unitary rotations, we say our standard form must have this structure. In this form $H$ is real, because $a_{IY} = a_{YI}= a_{YX}=a_{ZY}=a_{YZ} =0$, and since this is a necessary condition for $H$ to be a $Z$-matrix we demand that any further transformation we perform preserves this realness. 

We would now like to add some additional structure to our standard form. Before we can do this, we must handle some simple special cases. 

The first special case is when $S=P=0$. In this case we know $H$ is stoquastic, because one can always apply an appropriate choice of permutations and sign-flips on $\beta$ so that $a_{XX}' \leq -\vert a_{YY}' \vert$.

As a second special case, if $\beta =0$ then $H$ is stoquastic, since $S$ and $P$ may be freely rotated independently. 

Finally, as a third special case, if $a_{ZZ} = a_{XX} =0$, $a_{YY} \neq 0$, at least one of  $a_{XI}$ and $a_{IX}$ is non-zero, and at least one of  $a_{ZI}$ and $a_{IZ}$ is non-zero, then $H$ is not stoquastic. This is true by the following reasoning. Any transformation which makes $H$ a $Z$-matrix must satisfy $a_{XX}' \leq -\vert a_{YY}' \vert$ so that $a_{XX}' \neq 0$ or $a_{YY}' =0$. If $a_{YY}' \neq 0$ then $a_{XX}' \neq 0$ and since $\beta$ is rank-1, $a_{XY}' \neq 0$ and $a_{YX}' \neq 0$, which would disqualify $H$ from being a $Z$-matrix. Therefore any transformation which makes $H$ a $Z$-matrix must set $a_{YY}'=0$. However any such transformation will set at least one other $a_{\bullet Y}'$ or $a_{Y \bullet}'$ to be non-zero. Therefore $H$ is not stoquastic.

We may proceed under the assumption that these special cases have been handled, and argue that we can always put $S$, $P$ and $\beta$ into a standard form with the following properties, namely at least one $S$ or $P$ is non-zero and $\beta \neq 0$. Furthermore we may assume that $a_{ZZ} \geq a_{XX} \geq 0$. This follows from the fact that $a_{ZZ}$ and $a_{XX}$ can always be made non-negative by perform sign flips on the eigenvalues of $\beta$ and dumping any extra sign into $a_{YY}$ so as to preserve the determinant, furthermore any permutation between the $a_{XX}$ and $a_{ZZ}$ eigenvalues will preserve the standard form. Finally we may assume $H$ can be normalized so that $a_{ZZ} = 1$. This follows from the fact that the exclusion of the second and third special cases allows us to always put $S$, $P$, and $\beta$ into a standard form such that $a_{ZZ}>0$, since if $a_{ZZ}=a_{XX}=0$, but $a_{YY} \neq 0$ then either $a_{IX}=a_{XI}=0$ or $a_{IZ}=a_{ZI}=0$, in which case a permutation can always be performed that preserves the standard form but sets $a_{ZZ}' \neq 0$. 

Having put $S$, $P$, and $\beta$ into standard form with the above structure, we wish to know if there exists a pair of $SO(3)$ rotations $O_1$ and $O_2$ such that $\beta' =O_1 \beta O_2^T$, $S'=O_1 S$ and $P'=O_2 P$ where $\beta',S',P'$ are associated with a \zmatrix{} Hamiltonian. As discussed above, any such transformation must preserve the realness of the Hamiltonian. If it is not the case that either $a_{XI}=a_{IX}=0$ or $a_{ZI}=a_{IZ}=0$ then the only transformations which preserve this realness are $SO(3)$ rotations in the $X$-$Z$ subspace combined with reflections in the $X$- or $Z$-axis, given by 
\begin{align}\label{eq:sfRot}
O = \left(\begin{matrix} \cos(\theta)& 0&  \sin(\theta) \\ 0 &\gamma  &0  \\  -\gamma \sin(\theta)&0  &\gamma \cos(\theta)  \end{matrix}\right), && \gamma = \pm 1
\end{align} (If it is the case that either $a_{XI}=a_{IX}=0$ or $a_{ZI}=a_{IZ}=0$  then there is an additional transformation which can be performed, see the extra step in Appendix \ref{inequalityAppendix}.)


Recalling Proposition \ref{stoquasticCondition}, three inequalities must be satisfied in order for the rotated $H$ with $\beta',S',P'$ to be a \zmatrix{}. As we show in Appendix \ref{inequalityAppendix} these inequalities can be re-expressed as systems of two-variable polynomials which are at most quadratic in either variable, so analytic solutions to their roots can be constructed and solutions can be found using graphical methods. However, the complexity of these inequalities makes it unclear whether their interest goes beyond that of numerically finding a set of local basis changes.

\section{The Stoquastic Problem for 2-local Hamiltonians}\label{polyAlgoXYZ}

Consider a general $n$-qubit 2-local Hamiltonian
\begin{equation}
H=\sum_{uv \in E} H_{uv}+ \sum_{v\in V} H_v,
\end{equation}
where $u$ and $v$ are the vertices of the interaction graph $G=(V,E)$ of the Hamiltonian $H$. Each vertex $v \in V$ corresponds to a qubit so that $H_v$ are the 1-local terms (corresponding to $S$ and $P$ in the two-qubit case). Each edge $e=uv \in E$ is represented as $H_{uv}$ and can be associated with the $3 \times 3$ matrix $\beta^{uv}$ for that edge. 

One wants to determine when such a Hamiltonian is stoquastic. Recalling Definition \ref{def:stoq}, we want to find both a basis as well as a local decomposition $H= \sum_i D_i$ so that each $D_i$ is a Z-matrix acting on at most 2 qubits.
We begin with two simple cases. In the first case we consider how to find a local decomposition so that each $D_i$ is a $Z$-matrix, ignoring any basis transformations. Next we consider a family of natural Hamiltonians which have identical two-local terms on a bipartite graph. In this case the problem of finding an appropriate decomposition and basis transformation becomes much simpler. The final case we consider is the XYZ Heisenberg Hamiltonian, where we show that there exists a polynomial time algorithm for deciding if such a Hamiltonian is stoquastic. This constitutes our most significant result.

\subsection{Finding a Decomposition in a Fixed Basis}
Even given a fixed basis it remains non-trivial to identify if a decomposition of the form $H= \sum_i D_i$ exists. 

We can define the set of 2-local symmetric $Z$-matrices as the polyhedral cone $\mathcal{C}_n$ which is generated by the extremal vectors of each of the 2-qubit cones $\mathcal{C}_2$, which are now embedded in a higher dimensional $n$-qubit Hilbert-Schmidt space. 

Given a fixed basis, the cone $\mathcal{C}_n$ corresponds to those $H$ which admit such a decomposition $H=\sum_i D_i$ so that each $D_i$ is a Z-matrix acting on at most 2 qubits, and we are interested in deciding our Hamiltonian lies in that cone. A necessary condition is that the Hamiltonian is real, however this is straightforward to check. So going forward we assume that $H$ is real and want to decide whether $H \in \mathcal{C}_n$.

It is important to note that for any $v$, $\mathcal{C}_2^{uv}$ and $\mathcal{C}_2^{vu'}$ intersect: for any $u,u' \neq v$, the extremal vectors obey ${\bf o}_3^{uv}+{\bf o}_4^{uv}={\bf o}_5^{vu'}+{\bf o}_6^{vu'}$. This element in the intersection of $\mathcal{C}_2^{uv}$ with $\mathcal{C}_2^{vu'}$ is the 1-local $-X_v$ term.  This non-empty intersection shows the freedom in the decomposition $D_i$ for a fixed basis. The parsimonious strategy below shows how to find a decomposition $D_i$ of a $n$-qubit Hamiltonian $H$ in terms of 2-local terms each of which is a $4 \times 4$ symmetric Z-matrix or decide that it does not exist.

Let $H_u^X$ be the 1-local term proportional to $X_u$ and similarly $H_u^Z$. First, note that  a real $H$ can only be in $\mathcal{C}_n$ when for all $u$, $H_u^X \leq 0$. This follows directly from Proposition \ref{stoquasticCondition} demanding that $a_{IX}, a_{XI} \leq 0$. Hence we assume this to be the case (otherwise we conclude that no decomposition exists in the given basis).

{\em Efficient Parsimonious Strategy}: \\
Repeat the following for all edges $uv$ in $H$.
\begin{enumerate}
\item Given the current Hamiltonian $H$, pick a pair of vertices $u$ and $v$ and consider $h(\alpha, \beta)=H_{uv}+H_{u}^Z+H_v^Z-\alpha X_u-\beta X_v$ which includes {\em all} current single-qubit terms which act on vertices $u$ and $v$ and wlog $\alpha \geq 0,\beta \geq 0$.

Find the minimal $\alpha_{\rm min} \leq \alpha, \beta_{\rm min} \leq \beta$ such that $h(\alpha_{\rm min},\beta_{\rm min})\in \mathcal{C}_2$, or decide that $h(\alpha,\beta) \notin \mathcal{C}_2$. In the latter case, decide that $H \not \in \mathcal{C}_n$ and exit. Note that when $h(\alpha,\beta) \notin \mathcal{C}_2$ then $\forall \alpha' \leq \alpha, \beta' \leq \beta, \; h(\alpha',\beta') \notin \mathcal{C}_2$ since $a_{IX} \leq -|a_{ZX}|$ and $a_{XI} \leq -|a_{XZ}|$, aka it is easier to satisfy the stoquastic condition for large $\alpha,\beta$, more negative 1-local $X$-terms. Hence the goal is to use as little $\alpha$ and $\beta$ as possible, be parsimonious, so that a lot of $-X_u$ and $-X_v$ will be left over for another edge. This is then clearly an optimal strategy.
\item Define the left-over Hamiltonian $H_{\rm left-over}=H-h(\alpha_{\rm min},\beta_{\rm min})$ and repeat the previous step with the left-over Hamiltonian as the current Hamiltonian.
\item Stop the iteration when you have done all the edges $uv$. Either the left-over Hamiltonian is now 0 or it is the sum of 1-local $X$-terms, each of which has a negative sign, hence $H \in \mathcal{C}_n$.
\end{enumerate}

This problem is also discussed in \cite{Marvian2018}, where they argue that it can in general always be solved using linear programming.


\subsection{Uniform Bipartite Graphs}

Let us next consider a particularly simple example for which a resolution of the two-qubit problem analyzed in section \ref{twoQubit} suffices for deciding if the Hamiltonian is stoquastic.

\begin{proposition}Let the interaction graph $G=(V,E)$ of a Hamiltonian $H$ be bipartite, i.e. $V=V_A \cup V_B$, $V_A \cap V_B = \emptyset$, and there are only edges $uv \in E$ with $u \in V_A, v \in V_B$. Furthermore, $H=\sum_{uv \in E} h_{uv}$ where $h_{uv}$ acts with both one and two local terms on sites $u$ and $v$, and $h_{uv}=h$ for all $u \in V_A$, $v\in V_B$. If the two-qubit Hamiltonian $h$ is stoquastic, then $H$ is stoquastic. If $h$ is not stoquastic then $H$ will not be stoquastic under any local basis change which acts identically on all qubits in a partition.
\label{prop:bip}
\end{proposition}

Bipartite graphs include linear arrays, square lattices, cubic lattices and hexagonal lattices, all of which are very natural structures to consider.

\begin{proof}
If there exists a $U_A \otimes U_B$ such that $(U_A \otimes U_B) h (U_A \otimes U_B)^{\dagger}$ is a \zmatrix{}, then apply this rotation $U_A$ to all $u \in V_A$ and apply $U_B$ to all $v \in V_B$ and $H'=\sum_{uv \in E} (U_A \otimes U_B) h_{uv}(U_A \otimes U_B)^{\dagger}$ with $(U_A \otimes U_B) h_{uv}(U_A \otimes U_B)^{\dagger}$ a \zmatrix{}. 

Suppose $h$ were not stoquastic, but there existed a pair of unitaries $U_A$ and $U_B$ such that a decomposition $H= \sum_{uv \in E} D_{uv}$ exists where $D'_{uv}=(U_A \otimes U_B) D_{uv}(U_A \otimes U_B)^{\dagger}$ is a \zmatrix{}. $h'_{uv} =(U_A \otimes U_B) h_{uv} (U_A \otimes U_B)^{\dagger}$ cannot be a \zmatrix{}, and so $D_{uv} \neq h_{uv}$. However both $D_{uv}$ and $h_{uv}$ must share the same purely two local parts. As such, $D'_{uv}$ must differ from $h'_{uv}$ by its 1-local terms, and in particular one or both of the $-X_u$ and $-X_v$ terms. Since $h'_{uv}$ is not stoquastic, $D'_{uv}$ must have more support on either $-X_u$ or $-X_v$ than $h'_{uv}$. Suppose wlog that $D'_{uv}$ has more support on $-X_u$, then, by the parsimonious reasoning outlined earlier, there must exist some $D'_{ux}$ which has less support on $-X_u$ than $h'_{ux}$ does, and consequently $D'_{ux}$ is not a \zmatrix{}, leading to a contradiction.
\end{proof}

\subsection{Efficient Algorithm for XYZ Heisenberg Models}
\label{sec:proof-XYZ}

In this section we present and prove our main result, namely that there is an efficient algorithm to decide if an XYZ Heisenberg Hamiltonian is stoquastic.

We define an XYZ Heisenberg Hamiltonian as an $n$ qubit Hamiltonian of the form 
\begin{align}
H =& \sum_{uv} H_{uv}, \\
H_{uv} =& a_{XX}^{uv} X_u X_v+ a_{YY}^{uv} Y_u Y_v+ a_{ZZ}^{uv} Z_u Z_v.
\end{align}
The interaction graph of the Hamiltonian may be of any form.
Our main theorem is as follows.
\begin{theorem}
There exists a constructive algorithm, which runs in $O(n^3)$ time, that decides if an ${\rm XYZ}$ Hamiltonian H is stoquastic. 
\end{theorem}

\subsubsection{Preliminaries and Proof Outline}
Our proof consists of first describing an algorithm, and secondly proving the correctness of the algorithm.

The general aim of the algorithm is to find a set of single qubit Clifford rotations which transforms the Hamiltonian into a Z-matrix. If such a set is found, then our Hamiltonian is stoquastic.  We claim that if our algorithm cannot find such a set, then no such set exists, and furthermore that this implies that no local unitary rotation exists which makes $H$ a Z-matrix, and therefore the Hamiltonian is not stoquastic.

We find that the most natural way to represent the problem is as a matrix-weighted graph, with each vertex corresponding to a qubit, and each edge corresponding to an interaction term $H_{uv}$. The matrix weight of the edge corresponding to the term $H_{uv}$ is the $\beta^{uv}$ matrix discussed in Section \ref{convRep}, i.e.  a $3\times 3$ real diagonal matrix whose diagonal entries correspond to the coefficients of the XX, YY and ZZ terms. As discussed in Section \ref{convRep}, single qubit unitary rotations on qubits $u$ and $v$ correspond to $SO(3)$ rotations $O_u$ and $O_v$ acting by left and right matrix multiplication: $O_u \beta^{uv} O_v^T$. In this picture single-qubit Clifford rotations correspond to signed permutations with determinant 1, and so the goal of our algorithm is to assign signed permutations $\sPi_u$ to each vertex so that they transform all of the weights $\sPi_u \beta^{uv} \sPi_v^T$ into a form which satisfies the stoquasticity conditions given by proposition \ref{stoquasticCondition}. We refer to such an assignment $\{ \sPi_u \}$ as a \textbf{solution}.

Our algorithm relies on three important points. The first is that it is sufficient to consider only Clifford transformations, instead of all possible unitaries, so in our algorithm we only need to consider signed permutations. The second point is that any solution which transforms the Hamiltonian into a Z-matrix must preserve the diagonal form of the $\beta$ matrices. This follows from Proposition \ref{stoquasticCondition} and the fact that there are no single qubit terms. Thus an assignment of signed permutations $\{\sPi_u \}$, with $\tilde{\beta}^{uv}=\sPi_u \beta^{uv} \sPi_v^T$ is considered a solution if and only if
\begin{align}\label{diag}
\tilde{\beta}^{uv}_{ij}=\delta_{ij} \tilde{\beta}^{uv}_{ii} \; \mbox{, ie $\tilde{\beta}$ remains diagonal,}
\end{align}
and
\begin{align}\label{sign}
\tilde{\beta}^{uv}_{11}\leq& - \vert \tilde{\beta}^{uv}_{22} \vert 
\end{align}
the latter of which is equivalent to
\begin{align}
\label{perm}
[\tilde{\beta}^{uv}_{11}]^2 \geq & [\tilde{\beta}^{uv}_{22}]^2  \mbox{ and} \\
\tilde{\beta}^{uv}_{11} \leq  & 0 \label{negative}
\end{align}

The third point is that when an edge has a matrix weight $\beta$ which has a rank greater than 1, then in order to preserve the diagonal form of the matrix, the signed permutation acting from the left must be the same as to the signed permutation acting on the right, up to a difference of signs.

Any solution of signed permutations $\{\sPi_u\}$ admits a decomposition into a specification of permutations $\{\Pi_u \in S_3\}$, and signed diagonal matrices $\{R_u  = {\rm diag}(\pm 1, \pm 1, \pm 1) \}$ so that $\{\sPi_u = R_u \Pi_u \}$. 
The algorithm thus breaks up into two parts. The goal of the first part is to restrict the possible permutations $\{\Pi_u\}$  to those which satisfy condition \ref{perm}. The goal of the second part is to determine for which of those possible permutations an assignment of signs $\{R_u\}$ can be made so that condition \ref{negative} is also satisfied.

The general strategy of the first part of the algorithm is to `quotient out' clusters of vertices which are connected by \highrank\ edges, since the permutations applied to those vertices need to be all identical. 
One is left with a graph composed only of \lowrank\ edges. The goal is then to find permutations at the vertices which transform away any rank-1 $\beta$ matrix proportional to a YY term in the Hamiltonian, since such lone (not accompanied by XX) YY terms are forbidden by inequality \ref{perm}. 
By temporarily ignoring particular vertices, this task can be reduced to deciding if there exists an exact satisfying assignment to a classical Ising problem, which is an XOR-SAT problem and can be decided straightforwardly. The solution to this Ising problem translates back to the original graph in the form of sets of compatible permutation assignments. 

We then proceed to the second part of the algorithm, in which we determine for which (candidate) selections of permutations one can choose signs so that condition \ref{negative} is also satisfied. 

For any given choice of permutations, deciding the existence of such appropriate signs again reduces to deciding if there exists an exact solution to a classical Ising problem, and the number of such Ising problems one needs to check is polynomially bounded.

For the sake of clarity we present the definitions, the algorithm, and the proof of correctness separately, but the reader may find it useful to read these sections in parallel.

\subsubsection{Definitions}
\begin{definition}[\basegraph]
Given a Hamiltonian, the \textbf{\basegraph}\ is a matrix weighted graph, with each vertex corresponding to a qubit, and each edge $(u,v)$ weighted by the $\beta^{uv}$ matrix of the $H_{uv}$ term in the Hamiltonian.
\end{definition}
In this algorithm we will be considering weighted interactions graphs of $n$ qubit XYZ Heisenberg Hamiltonians.


\begin{definition}[\lowrank\ and \highrank\ edge]
An edge in a \basegraph\ is \textbf{\lowrank}\ if its matrix weight has a rank 1. Conversely, an edge in a \basegraph\ is \textbf{\highrank}\ if its matrix weight has a rank greater than 1.
\end{definition}

\begin{definition}[\highrank\ connected component]
Consider a \basegraph\ $G$. Remove all \lowrank\ edges. One is left with a set of distinct connected components which are composed entirely of \highrank\ edges. Each of these connected components is a subgraph of $G$ which we call a \highrank\ connected component $\Gamma$. 
\end{definition}
Note that a \highrank\ connected component may contain a single vertex with no edges, in the case where some vertex is connected to only \lowrank\ edges. Note also that every vertex in the weighted interaction graph belongs to exactly one \highrank\ connected component.

\begin{definition}[\quotientGraph]
Given a \basegraph\ of an XYZ Heisenberg Hamiltonian, a \textbf{\quotientGraph} is a multi-graph with labelled edges which is constructed as follows. For each \highrank\ connected component in the \basegraph\, populate the \quotientGraph\ with a corresponding vertex. Since every vertex in the \basegraph\ belongs to exactly one \highrank\ connected component, every vertex in the \basegraph\ has a corresponding representative in the \quotientGraph . For every pair of vertices in the \basegraph\ connected by a \lowrank\ edge, connect their corresponding representative vertices in the \quotientGraph\ by an edge. Label this new edge either $1$, $2$, or $3$, corresponding to the index $i$ for which $\left[\beta \right]_{ii} \neq 0$. 
\end{definition}

This construction effectively quotients out all the \highrank\ connected components. It is however a multi-graph because there may be multiple \lowrank\ edges connecting the vertices in a pair of \highrank\ connected components.

\begin{definition}[single-label vertex]
Given a \quotientGraph\, a \textbf{single-label vertex} is a vertex incident to edges of a single label.
\end{definition}

\begin{definition}[single-label connected component]
Consider a \quotientGraph\ $G_Q$ with all vertices removed which are not single-label, call this graph $G_Q^-$. For each subset of vertices $S_Q$ associated with a connected component of $G_Q^-$, the subgraph $\Gamma_Q$ of $G_Q$ induced \footnote{A subgraph of a graph $G=(V,E)$ induced by a subset of vertices $S \subseteq V$ is defined as a graph whose vertex set is $S$ and whose edge set is all edges in $E$ that have both endpoints in $S$.} by $S_Q$ is called a \textbf{single-label connected component}. 

\end{definition}

Note that because every vertex in a given single-label connected component is incident to edges of a single label $i$, every edge in the single-label connected component must have a common label $i$, and we may say that a single-label connected component is labelled $i$. Every single-label vertex in the \quotientGraph\ belongs to exactly one single-label connected component. 


\begin{definition}[heterogeneous \quotientGraph]
Given a \quotientGraph , the \textbf{heterogeneous \quotientGraph} is a copy of the \quotientGraph\ modified in the following way. For each single-label connected component, with label $i$, remove all vertices in the single-label connected component from the \quotientGraph , and connect every vertex in the boundary of the single-label connected component to every other vertex in the boundary, as well as to itself, by an edge labelled by $i$. 
\end{definition}

By construction, every single-label vertex is removed from the heterogeneous \quotientGraph . Therefore every vertex in this new heterogeneous \quotientGraph\ will be incident to edges labelled by more than one label.

\begin{definition}[Admissible Permutations]
The set of \textbf{admissible permutations} $\tilde{S}$ of a \highrank\ connected component $\Gamma$ is a subset of permutations $\tilde{S} \subseteq S_3$  such that for each permutation $\Pi \in \tilde{S}$, when applied identically to every vertex in $\Gamma$, transforms every edge weight in $\Gamma$ to a form which satisfies condition \ref{perm}. Since every vertex $u$ in both the \quotientGraph\ and the heterogeneous \quotientGraph\ is associated with a \highrank\ connected component, each such vertex is associated with a set of \textbf{admissible permutations} $\tilde{S}_u$.
\end{definition}


\begin{definition}[Compatible Set of Permutation Assignments]
Given a \quotientGraph , or a heterogeneous \quotientGraph\ , a \textbf{compatible set of permutation assignments} is an assignment of sets of possible permutations to each vertex in the graph: $A=\{\Sigma_u \subseteq \tilde{S}_u \}$ satisfying the following property. For every edge with label $i$ and end vertices $u$ and $v$, there exists a number $j \neq 2$ such that every permutation $\Pi$ in $\Sigma_u$ or $\Sigma_v$ satisfies $\Pi(i) =j$. 

\end{definition}

\begin{definition}[partitioned \basegraph]
Given a \quotientGraph , a \basegraph\, and a compatible set of permutation assignments: $A=\{\Sigma_u \subseteq S_3 \}$, the \textbf{partitioned \basegraph} is a copy of the \basegraph\ modified in the following way. For each edge in the \quotientGraph\ labelled $i$, with end vertices $u$ and $v$, any permutation $\Pi$ in $\Sigma_u$ or $\Sigma_v$ will map $i$ to specific label $\Pi(i)=j \neq 2$. If $j \neq 1$, remove the corresponding edge from the \basegraph . 
\end{definition}

Note that the partitioned \basegraph\ may or may not be fully connected.


\begin{definition}[Ising Partition]
Given a connected component $\Gamma_W$ of a partitioned \basegraph\, and an assignment of permutations $\{\Pi_u\}$ to each vertex $u$ in the connected component, an \textbf{Ising Partition} is a copy of $\Gamma_W$, with each edge labelled by the value $ -\textrm{\rm sign}([\Pi_u \beta^{uv} \Pi_v^T]_{11})$.  \label{isingSign}
\end{definition}

The Ising partition can be thought of as a classical Ising model, since it is a graph whose edges are labelled by $+1$ or $-1$. One can ask if there exists an assignment of $+1$ or $-1$ to each vertex in the graph so that for every edge the product of its neighbouring vertices is equal to the sign associated with that edge. If such an assignment exists, we call it an exact satisfying solution to the classical Ising problem. Note that this problem is distinct from the generic classical Ising problem, which asks for a minimizing assignment instead of an exact assignment. It is straightforward to decide if such an exact satisfying solution exists. Simply assign a sign to one vertex, and propagate that choice through the graph, according to the edge constraints, until you either encounter a term which is not satisfied under this procedure, or have successfully assigned signs to every vertex. Since the procedure runs over all edges it runs in $O(n^2)$ time. If a solution exists, then two solutions must exist, since one can simply flip every sign and still have solution.

\subsubsection{The Algorithm}
The algorithm presented here takes as input an $n$ qubit XYZ Heisenberg Hamiltonian, and outputs either true, indicating that the Hamiltonian is stoquastic, or false, indicating that the Hamiltonian is not stoquastic. In the event that the algorithm outputs true, the algorithm also provides a set of local unitary rotations which transform the Hamiltonian into a Z-matrix. In this sense it is a constructive algorithm.

\noindent\rule{0.5 \textwidth}{0.4pt}

\noindent\textbf{Part 1 of the Algorithm}

\begin{wrapfigure}{r}{.19 \textwidth}
\includegraphics[width=.19 \textwidth]{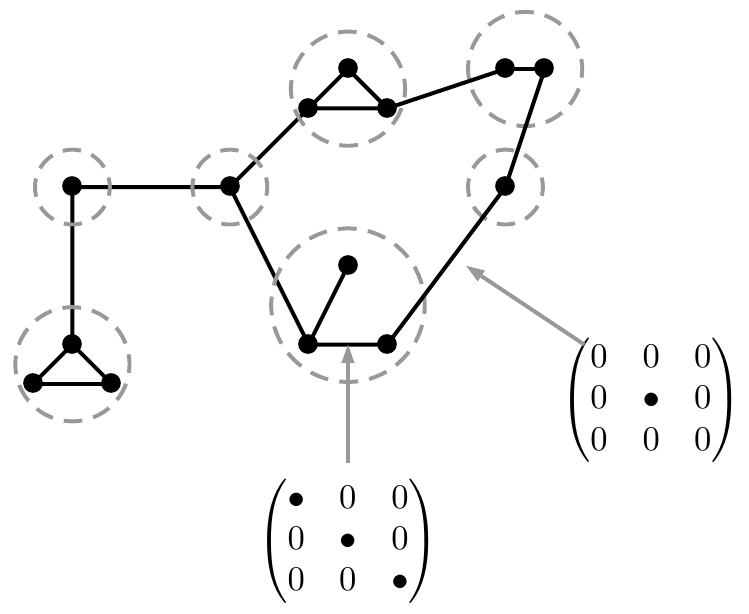}
\end{wrapfigure}
\textbf{Step 1:} Construct the \basegraph\ of the Hamiltonian. $O(n^2)$

%

\textbf{Step 2:}  Identify all \highrank\ connected components of the \basegraph\, and determine their admissible permutations. $O(n)$



\begin{wrapfigure}{l}{.21 \textwidth}
\includegraphics[width=.21 \textwidth]{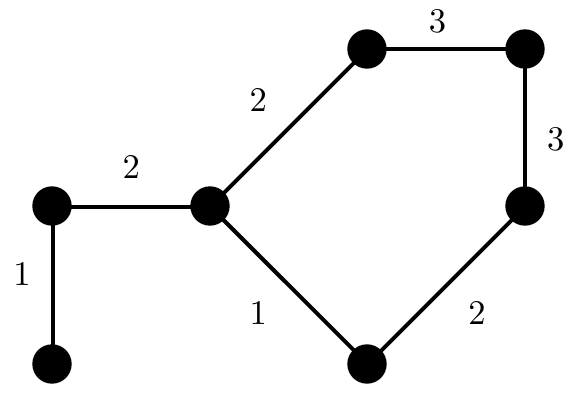}
\end{wrapfigure}

\textbf{Step 3:} Construct the \quotientGraph\ of the \basegraph . $O(n^2)$

\textbf{Step 4:} If there are any vertices in the \quotientGraph\ which connect to at least three edges, each with a different label, then return \textbf{false}. If not, then every vertex in our graph connects to edges which are labelled by at most two distinct labels. $O(n^2)$

\textbf{Step 5 edge case:} If the whole \quotientGraph\ is itself a single-label connected cluster, then and only then will the heterogeneous \quotientGraph\ be empty.  If so proceed to the step 9 edge case. $O(n)$

\begin{wrapfigure}{l}{.2 \textwidth}
\includegraphics[width=.2 \textwidth]{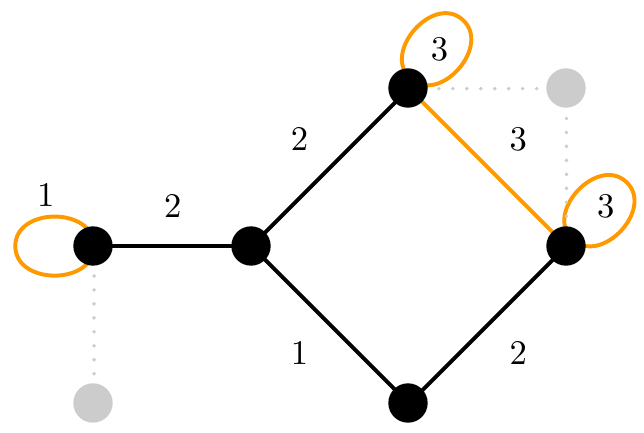}
\end{wrapfigure} 
\textbf{Step 5:} Construct the heterogeneous \quotientGraph\ from the \quotientGraph . Since every vertex in the \quotientGraph\ connects to edges which are labelled by at most two distinct labels, every vertex in the  heterogeneous \quotientGraph\ must be incident to edges with exactly two distinct labels. Let us denote a vertex incident to edges labelled $i$ and $j$ as an $ij$-vertex. $O(n)$

\textbf{Step 6:} A classical Ising problem may be specified by a graph with $+1$ and $-1$ labels on the edges, \begin{wrapfigure}{l}{.2 \textwidth}
\includegraphics[width=.2 \textwidth]{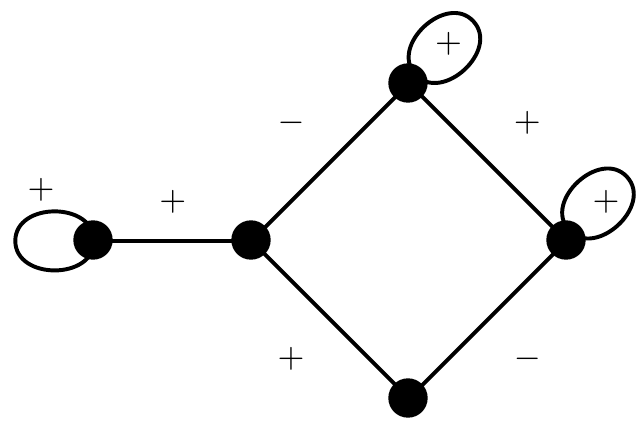}
\end{wrapfigure}
and the solution corresponds to an assignment of $+1$ and $-1$ to the vertices so that the product of an edge with its vertex assignments produces $+1$.  Using the heterogeneous \quotientGraph , construct an Ising problem in the following way. Copy the graph structure of the heterogeneous \quotientGraph . Label the edges of this new graph using the following prescription.  $O(n^2)$

\begin{itemize}
\item  For each edge in the heterogeneous \quotientGraph\ labelled $2$, with one incident $12$-vertex and one incident $23$-vertex, label the corresponding edge in the Ising problem with a $-1$.
\item Label all other edges in the Ising problem with a $+1$.
\end{itemize}

\begin{wrapfigure}{r}{.22 \textwidth}
\includegraphics[width=.22 \textwidth]{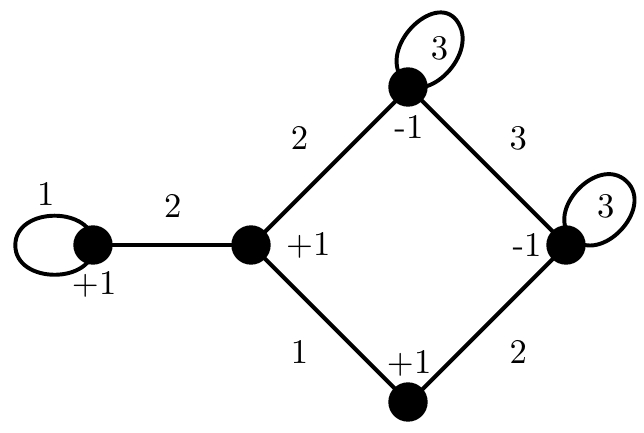}
\end{wrapfigure}
\textbf{Step 7:} If there does not exist an exactly satisfying solution to this classical Ising problem then return \textbf{false}. If there does, then there must exist two possible solutions $I_1$ and $I_2$.  $O(n)$

\begin{wrapfigure}{r}{.22 \textwidth}
\includegraphics[width=.22 \textwidth]{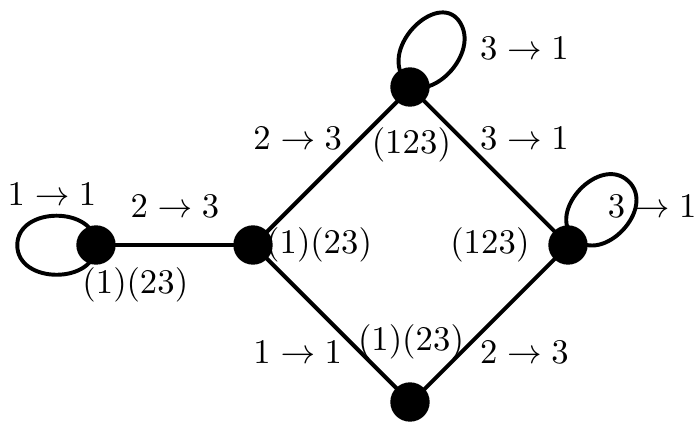}
\end{wrapfigure}
\textbf{Step 8:} A solution to the Ising problem described above uniquely specifies an assignment of permutations to the vertices of the heterogeneous \quotientGraph , when it is combined with the information about the edges to which each of those vertices connect. This is most intuitively illustrated by figure  \ref{hexagon}, and a concise description of the unique assignment is given by table \ref{rubric}.
\begin{table}[h!]
\begin{center}
\begin{tabular}{|c|c|c|}
\hline
Ising solution:  & +1 & -1  \\
\hline 
 13-vertex & (1)(2)(3) & (13)(2)\\
 12-vertex & (1)(23) & (132)\\
 23-vertex & (12)(3) & (123)\\
 \hline
\end{tabular}
\caption{Translation between permutation and Ising solution}\label{rubric}
\end{center}
\end{table}
 
\begin{figure}
\begin{center}
\includegraphics[scale=.6]{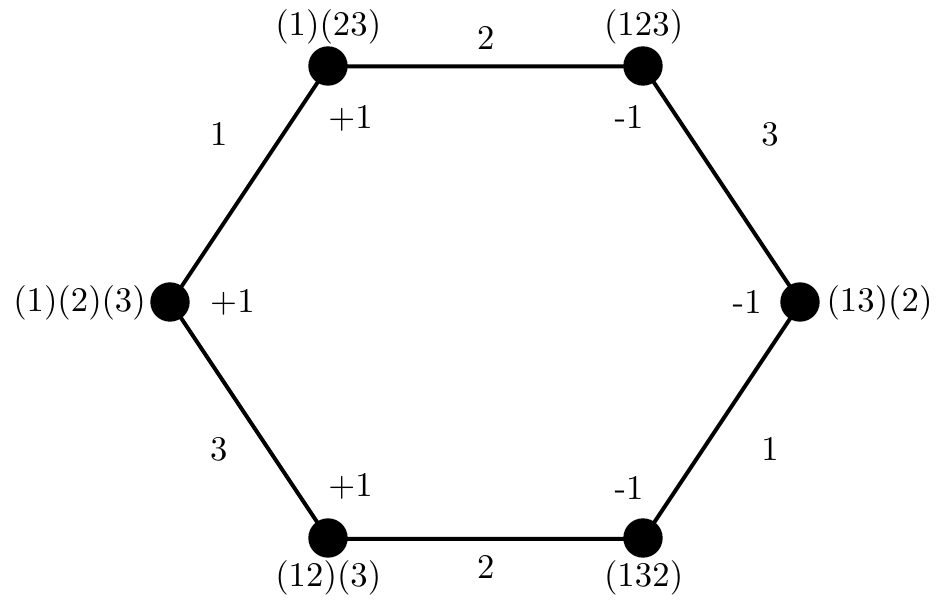}
\end{center}
\caption{A complete illustration of which permutations to apply to a vertex, given a solution to an Ising problem on the heterogeneous \quotientGraph}\label{hexagon}
\end{figure}

 Using this rubric, and the two solutions $I_1$ and $I_2$, construct two possible assignments of permutations to the vertices of the heterogeneous \quotientGraph\ $A^h_1 = \{\Sigma_{u}^{h}=\{\Pi_u^1\}\}$ and $A^h_2 = \{\Sigma_{u}^h= \{\Pi_u^2\}\}$. If for either assignment, the permutation assigned to some vertex is not in the admissible permutations of that vertex, then discard that assignment. If both assignments are discarded in this way, return \textbf{false}. Any remaining $A^h$ constitute compatible sets of permutations assignments to the heterogeneous \quotientGraph , with each $\Sigma_u^h$ containing a single element. $O(n)$
 

\textbf{Step 9 preliminaries} We now wish to construct all compatible sets of permutation assignments for the \quotientGraph . Note that unlike in the case of the heterogeneous \quotientGraph , the permutations at the single label vertices may not be completely specified in a given compatible set of permutations assignment, so that some $\Sigma_u$ may have more than one element. 



\textbf{Step 9 edge case:} In the case where the whole \quotientGraph\ is a single-label connected component with label $i$, as is the case in the step 5 edge case, construct two compatible sets of permutation assignments $A_1$ and $A_2$. $A_1 =\{ \Sigma_u^1 \}$, and $\Sigma_u^1$ is all permutations $\Pi$ in the admissible permutations of vertex $u$ which satisfy $\Pi(i)=1$. $A_2 =\{ \Sigma_u^2 \}$, and $\Sigma_u^2$ is all permutations $\Pi$ in the admissible permutations of vertex $u$ which satisfy $\Pi(i)=3$. If either $A_1$ or $A_2$ contains an element $\Sigma_u$ which is empty, then discard it. If both $A_1$ and $A_2$ have been discarded, return \textbf{false}. $O(n)$

\begin{wrapfigure}{r}{.20 \textwidth}
\includegraphics[width=.20 \textwidth]{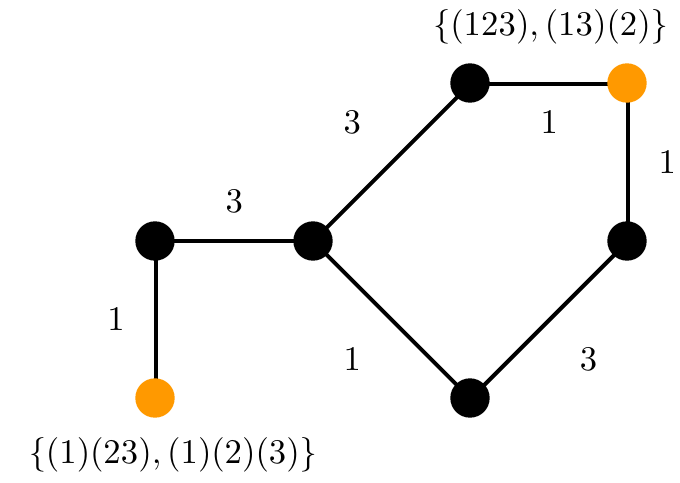}
\end{wrapfigure}
\textbf{Step 9:} Each compatible set of permutationa assignments to the heterogeneous \quotientGraph\ translates into an assignment of permutations to all of the non-single-label vertices in the \quotientGraph . This partially specifies up to two compatible sets of permutation assignments $A_1=\{\Sigma_u^1\}$ and $A_2=\{\Sigma_u^2\}$, where $\Sigma_u^1 = \Sigma_u^h \in A_h^1$ and $\Sigma_u^2 = \Sigma_u^h \in A_h^2$ are single element sets defined only on non-single label vertices. For each compatible set of permutation assignments $A_x$, perform the following procedure. For each single-label connected component in the \quotientGraph , with label $i$, choose a vertex $u$ from its boundary, which will have a permutation $\Pi_u^x$ assigned to it. There exists only one other permutation $\Pi' \neq \Pi_u^x$ such that $\Pi_u^x(i) = \Pi'(i)$. For each vertex $v$ in the single-label connected component, let $\Sigma_v^x$ be all permutations in $\{\Pi', \Pi_u^x \}$ which are also in the admissible permutations of vertex $v$. If either $A_1$ or $A_2$ contains an element $\Sigma_u$ which is empty, then discard it. If both $A_1$ and $A_2$ have been discarded, return \textbf{false}. $O(n)$

\textbf{Conclusion of Part 1 of the Algorithm:}
We have now constructed up to two compatible sets of permutation assignments to the vertices in the \quotientGraph . Together these sets represent all possible assignments of permutations to the vertices which satisfy conditions \ref{perm} and \ref{diag}. The aim of the second part of the algorithm is to determine whether or not there exists any assignments of permutations from these sets which also admit an assignment of signs at every vertex $\{R_u = {\rm diag}(\pm 1, \pm 1, \pm 1) \}$ so that condition \ref{negative} is satisfied.

\noindent\rule{0.5 \textwidth}{0.4pt}

\noindent\textbf{Part 2 of the Algorithm}

\textbf{Step 10: } Let $sol = \{\sPi_u = R_u \Pi_u\}$ store the potential solution. 
For each set of compatible permutation assignments $A_x = \{\Sigma_u^x\}$, $x \in {1,2}$, from part 1, perform steps 11 through 13, discarding those $A_x$ which fail.



\begin{table}[h!b]
\begin{tabular}[t]{@{} m{.25 \textwidth}  m{.25 \textwidth}@{}}
 \hspace{.2cm} \textbf{Step 11:} Construct the partitioned \basegraph . $O(n^2)$  &
\includegraphics[scale=.47]{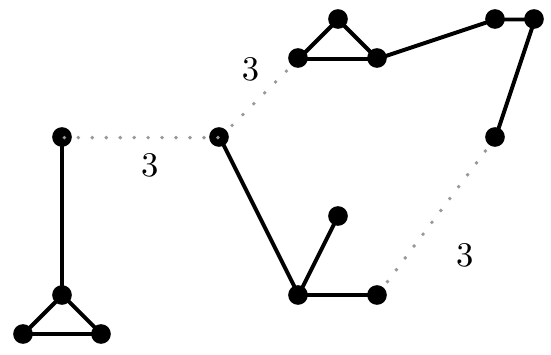}
\end{tabular}
\end{table}

\textbf{Step 12:} By construction, every single-label vertex $u$ in the \quotientGraph , with label $i$, for which the permutations $\Pi$ in $\Sigma_u^x$ satisfy $\Pi(i)=3$, will correspond to a unique connected component in the partitioned \basegraph . For each such vertex $u$, do the following.  For each permutation $\Pi \in \Sigma_u^x$, the permutation assignment for each vertex $v$ in the connected component is given by $\Pi$. Construct the corresponding Ising Partition and decide if it admits a solution. If yes, then choose a solution and let $\sPi_u = \textrm{diag}(\delta_u, 1 , \delta_u \det(\Pi)) \Pi$, and $\delta_u \in \{+1,-1\}$ is given by the solution to the Ising model. If neither permutation in $\Sigma_u^x$ admits a solution, then discard $A_x$. In the worst case there are $n$ single-label vertices, and in the worst case constructing and solving an Ising Partition is $O(n^2)$, so in the naive worse case the runtime is $O(n^3)$.

\begin{wrapfigure}{r}{.23 \textwidth}
\includegraphics[width=.23 \textwidth]{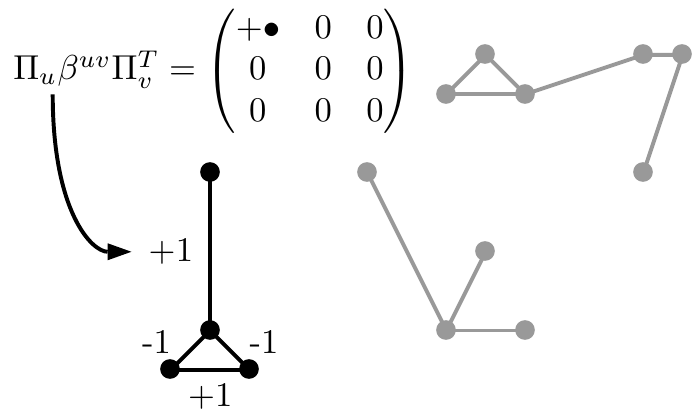}
\end{wrapfigure}
\textbf{Step 13:} For the remaining connected components $\Gamma_P$ in the partitioned \basegraph\, do the following.  For each vertex $v$ in $\Gamma_P$, choose a permutation $\Pi_v \in \Sigma_v^x$.  Equipped with this permutation assignment, construct the corresponding Ising partition and decide if it admits a solution. If not, then discard $A_x$. If yes, then choose a solution and let $\sPi_v = \textrm{diag}(\delta_v, 1, \delta_v \det(\Pi_v)) \Pi_v$ for all $v$ in $\Gamma_P$, and $\delta_v \in \{+1,-1\}$ is given by the solution to the Ising model. By similar reasoning to step 12, the worst case runtime is $O(n^3)$.



\textbf{Step 14:} If all $A_x$ have been discarded, return \textbf{false}. Otherwise return \textbf{true}, along with $sol$, which will have been completely specified for every vertex in the \basegraph .

\subsubsection{Correctness of the algorithm}

It is obvious that the algorithm is complete, as long as it is supplied an XYZ Heisenberg Hamiltonian, it will always either return true or false. We need only prove soundness. The proof of soundness of the algorithm is twofold. We must first show that when the algorithm returns false then $H$ is not stoquastic. Secondly we must show that when the algorithm returns true $H$ is stoquastic, and $sol$ corresponds to the rotation which transforms it into a Z-matrix.

\begin{lemma}[All Diagonal Matrix Weighted Graphs are Connected by Signed Permutations]\label{cliffDiag}
Given a set of diagonal $3\times 3$ matrices $\{ \beta^{uv} \}$ which are diagonal in some basis $b=\{\hat{e}_1, \hat{e}_2, \hat{e}_3 \}$. For any set of  $SO(3)$ rotations $\{O_u\}$ such that $ O_u \beta^{uv} O_v^T$ is also diagonal in basis $b$ for all $\beta^{uv}$, there exists a set of signed permutations $\{\sPi_u\}$ such that  $\sPi_u \beta^{uv} \sPi_v^T =  O_u \beta^{uv} O_v^T$
\end{lemma}

The proof of lemma \ref{cliffDiag} can be found in appendix \ref{permsandrefs}

\begin{lemma}[Single-Qubit Clifford Rotations Suffice] \label{cliffordsSuffice}
The ${\rm XYZ}$ Heisenberg Hamiltonian can be transformed into a $Z$-matrix by single-qubit unitary transformations if and only if it can be transformed into a $Z$-matrix by single-qubit Clifford transformations.
\end{lemma} 

\begin{proof}
 The ${\rm XYZ}$ Heisenberg Hamiltonian corresponds to a set of diagonal $3\times 3$ matrices $\{ \beta^{uv} \}$, and any transformation which transforms one into a $Z$-matrix corresponds to a set of $SO(3)$ rotations $\{O_u\}$ such that $ \tilde{\beta}_{uv} = O_u \beta^{uv} O_v^T$ satisfies conditions \ref{diag}, \ref{perm} and \ref{negative} in the basis $b$ for all $u$ and $v$. Therefore $\tilde{\beta}_{uv}$ is diagonal, and by lemma \ref{cliffDiag} there exists a set of set of signed permutations $\{\sPi_u\}$ such that $\sPi_u \beta^{uv} \sPi_v^T = O_u \beta^{uv} O_v^T$ also satisfies those conditions. These signed permutations may not have determinant 1, but there exist determinant 1 signed permutations $\sPi_u' = \textrm{diag}(1,1,\det{\sPi_u}) \sPi_u$ such that the matrices $\{\sPi_u'\beta^{uv} \sPi_v'\}$ also satisfy conditions \ref{diag}, \ref{perm} and \ref{negative} in the basis $b$ for all $u$ and $v$.  The other direction of the biconditional is trivial.
\end{proof}

\begin{lemma}[Permutations must be identical on \highrank\ connected components] If a set of signed permutations $\{\sPi_u=R_u \Pi_u\}$ is a solution, then for any \highrank\ connected component $\Gamma$, there exists a permutation $\Pi$ such that $\Pi_u=\Pi$ for every vertex $u$ in $\Gamma$.
\label{rigid}
\end{lemma}

\begin{proof}
If $\{\sPi_u \}$ is a solution, then $\sPi_u \beta^{uv} \sPi_v^T=\sum_m \beta^{uv}_{mm} R_{\Pi_u(m)}^u R_{\Pi_v(m)}^v \ket{\Pi_u(m)} \bra{\Pi_v(m)}$ is diagonal. This implies that if $\beta^{uv}_{mm} \neq 0$ then $\Pi_u(m)=\Pi_v(m)$. Since $\beta^{uv}_{mm} \neq 0$ for two of the three values of $m$, $\Pi_u=\Pi_v$. This holds for every pair of connected vertices in the connected components, hence for all $u \in \Gamma$, $\Pi_u=\Pi$ for some $\Pi$.
\end{proof}

\begin{lemma}[\quotientGraph\ permutation assignments are necessary] If there does not exist a compatible set of permutation assignments for the \quotientGraph , then $H$ is not stoquastic.
\label{compatPerm}
\end{lemma}

\begin{proof}
If $H$ is stoquastic, then by lemma \ref{cliffordsSuffice} there exists a set of signed permutations $\{\sPi_u = R_u \Pi_u\}$ which transform $H$ into a Z-matrix, and so for every edge weight $\beta^{uv}$, $\tilde{\beta}^{uv}=\sPi_u \beta^{uv} \sPi_v^T$ satisfies conditions \ref{diag}, \ref{perm} and \ref{negative}. Every vertex $x$ in the \quotientGraph\ is associated with a rigid connected component $\Gamma_x$ in the \basegraph . For every edge in the \quotientGraph , with label $i$, end vertices $x$ and $y$ and associated rigid connected components $\Gamma_x$ and $\Gamma_y$, there exists an edge in the \basegraph\ with end vertices $u \in \Gamma_x$ and $v \in \Gamma_y$ and edge weight $\beta^{uv}$ with rank 1 and $\beta^{uv}_{ii} \neq 0$. By lemma \ref{rigid}, there must exist permutations $\Pi_x$ and $\Pi_y$ such that for every vertex $u$ in $\Gamma_x$ and vertex $v$ in $\Gamma_y$, $\sPi_u = R_u \Pi_x$ and $\sPi_v = R_v \Pi_y$.  By conditions \ref{diag} and \ref{perm}, $\Pi_x(i) = \Pi_y(i)\neq 2$. Thus such a specification of permutations $A = \{ \Gamma_x=\{ \Pi_x\}\}$ corresponds to a compatible set of permutation assignments. The lemma follows contrapositively.
\end{proof}

\begin{lemma}[Single label paths are transformed uniformly]
Given a connected path of edges in a \quotientGraph\ which are identically labelled $i$, for any compatible set of permutation assignments $A=\{\Sigma_u\}$ there exists a label $j \neq 2$ such that for every vertex $u$ in the path, and for every $\Pi \in \Sigma_u$, $\Pi(i)=j$.
\label{path}
\end{lemma}

\begin{proof}
For any vertex $u$ in the path, there exists a label $j \neq 2$ such that for every $\Pi \in \Sigma_u$, $\Pi(i)=j$. By construction, every vertex in the path neighbouring $u$ satisfies this condition as well. The proof follows by induction.
\end{proof}



\begin{lemma}[Heterogeneous \quotientGraph\ permutation assignments are necessary] If there does not exist a compatible set of permutation assignments for the heterogeneous \quotientGraph , then there does not exist one for the \quotientGraph .
\label{hetCompatPerm}
\end{lemma}

\begin{proof}
Suppose there exists a compatible set of permutation assignments $A=\{\Sigma_u\}$ for the \quotientGraph . Let $A^h = \{\Sigma_u\}  \subset A$ be the assignment of permutations to those vertices $u$ which are also in the heterogeneous \quotientGraph . $A^h$ is a compatible set of permutation assignments for the heterogeneous \quotientGraph . To see that this is true, one only needs to note that the only edges that are added in the construction of the heterogeneous \quotientGraph\ are those which connect boundary vertices of connected components. For every pair of boundary vertices $v$ and $w$ of a connected component labelled by $i$, it follows from lemma \ref{path} that there exists a $j$ such that $\Pi(i)=j$ for all $\Pi \in \Sigma_v, \Sigma_w$, because there exists a path of edges labelled $i$ connecting $v$ and $w$ in the \quotientGraph . Thus $A^h$ is indeed a compatible set of permutation assignments, and the proof follows contrapositively.
\end{proof}

\begin{lemma}[The heterogeneous \quotientGraph\ translates faithfully to an Ising problem]
 Given a heterogeneous \quotientGraph\ with vertices which connect to edges with at most two labels. If there exists a compatible set of permutation assignments $A^h$, then there exists a solution to the Ising problem described in step 6, which when translated using table \ref{rubric} will produce $A^h$. 
 \label{isingPerm}
\end{lemma}

\begin{proof}
Suppose $A^h=\{\Sigma_u^h\}$ exists. For every edge in the heterogeneous \quotientGraph , with end vertices $u$ and $v$ and label $i$, there exists a $j \neq 2$ s.t. for all $\Pi_u \in \Sigma_u^h$, $\Pi_v \in \Sigma_v^h$: $\Pi_u(i) = \Pi_v(i) = j$. Let $u$ be an $ix$-vertex and $v$ be an $iy$-vertex. Then $\Pi_u(x) \neq 2$ and $\Pi_v(y) \neq 2$. Since $x\neq i$, there must exist a third index $k$ so that $\Pi_u(k) = 2$, this uniquely specifies two possible permutations, given by the rows of table \ref{rubric}, similarly for $\Pi_v$.  If $i \neq 2$ then both $\Pi_u$ and $\Pi_v$ must belong to the same column in table \ref{rubric}. If $i = 2$ then  $\Pi_u$ and $\Pi_v$ must belong to distinct columns in table \ref{rubric}. We see that any choice of assignments of permutations to the vertices satisfying the above restrictions specify an exact satisfying solution to the Ising model described in step 6, given by the labelling of the columns in table \ref{rubric}.
\end{proof}

\begin{theorem}
If the algorithm returns false, the Hamiltonian is not stoquastic.
\end{theorem}

\begin{proof}
There are six steps where the algorithm returns false. We shall address them individually.

\textbf{Step 4}: If a vertex in the \quotientGraph\ connects to at least three edges, each with a different label, then any assignment of permutation to that vertex will leave one label equal to 2. Thus no compatible set of permutation assignments is possible, and by lemma \ref{compatPerm} the Hamiltonian is not stoquastic.

\textbf{Step 7:} If there does not exist an exact satisfying solution to the Ising problem given in step 6, then by lemmas \ref{isingPerm}, \ref{hetCompatPerm}, and \ref{compatPerm} the Hamiltonian is not stoquastic.

\textbf{Step 8:} If both $A^h_1$ and $A^h_2$ contain permutations at some vertices which are not in the admissible permutations of those vertices, then they are not compatible sets of permutations assignments. Furthermore no others exist, by lemma \ref{isingPerm}, because an Ising problem has only two solutions. So by \ref{hetCompatPerm} and \ref{compatPerm} the Hamiltonian is not stoquastic.

\textbf{Step 9 edge case:} Since, in the step 5 edge case, there is a label-$i$ path between every vertex in the \quotientGraph , by lemma \ref{path} any set of compatible permutation assignments must be of the form $A_1$ or $A_2$. Thus if both are incompatible with the admissible permutations of the vertices, then by lemma \ref{compatPerm} the Hamiltonian is not stoquastic.

\textbf{Step 9:} By lemma \ref{path}, any assignment of possible permutations $\Sigma_v$ for a vertex $v$ in a single-label connected component with labelling $i$ must agree with the action $\Pi(i)$ of any permutation $\Pi \in \Sigma_u$ of a boundary vertex $u$ of that single-label connected component. If, for $A_x^h$ the admissible permutations of $v$ do not contain any permutations which agree in this way with the assignments of permutations to the boundary vertices, then no set of compatible permutation assignments exists on the \quotientGraph\ which corresponds to $A_x^h$. If both $A_1^h$ and $A_2^h$ are ruled out in this way, then by lemmas \ref{hetCompatPerm} and \ref{compatPerm} $H$ is not stoquastic.

\textbf{Step 14:}  A compatible set of permutation assignments $A_x$ can not correspond to a possible solution if, for all of the possible permutations producible by $A_x$, there does not exist an assignment of signs $\{R_u = \textrm{diag}(\pm1, \pm1, \pm1)\}$ satisfying condition \ref{negative}. Given a fixed choice of permutations $\{\Pi_u\}$, the question of whether there exists an assignment of signs is dependent only on the values $[\Pi_u \beta^{uv} \Pi_v^T]_{11}$ at each edge. Therefore the question of whether there exists an assignment of signs can be answered independently for each subset of vertices in the \basegraph\ which do not connect to any vertices outside of the subset with an edge satisfying $[\Pi_u \beta^{uv} \Pi_v^T]_{11} \neq 0$. Since for a given $A_x$, any choice of permutation assignment $\{\Pi_u \in \Sigma_u^x\}$ will, by definition, have the same action on all low-rank edges, it suffices to consider each connected component of the partitioned \basegraph\ independently. Steps 12 and 13 check each connected component of the partitioned \basegraph\ in this fashion. 

For each connected component of the partitioned \basegraph\, one must rule out all possible permutation assignments one could draw from $A_x=\{\Sigma_v^x\}$ and assign to the vertices in that connected component. In the case of step 12, each single-label vertex corresponds to its own connected component, and there are only two possible permutations to consider. In the case of step 13, which covers the other connected components of the partitioned \basegraph\, it may seem that there are an exponential number of possible assignments one could draw from. However each vertex $u$ in one of these connected components necessarily corresponds to either a non-single label vertex $v$ in the \quotientGraph , for which $\Sigma_v^x$ has only one element, or a single-label vertex $v$ in the \quotientGraph , labelled by $i$, for which any permutation $\Pi_v \in \Sigma_v^x$ satisfies $\Pi_v(i) = 1$. As such, the Ising model produced is independent of the choice of permutation assignment, so we choose a single representative.

By the reasoning outlined for the earlier steps, the only possible compatible sets of permutation assignments $A_x$ are the ones produced in part 1. If none of the $A_x$ correspond to a possible solution, then there is none and the Hamiltonian is not stoquastic.

\end{proof}

\begin{theorem}
When the algorithm returns true, along with a solution $sol$, the solution $sol=\{\sPi_u= R_u \Pi_u\}$ prescribes a set of single qubit unitary transformations $\{U_v\}$ such that $(\otimes_v U_v) H (\otimes_v U_v)^{\dagger}$ is a $Z$-matrix, and therefore $H$ is stoquastic.
\end{theorem}



\begin{proof}
As discussed in the preliminaries, single qubit unitary transformations correspond to $SO(3)$ rotations, and $H$ is stoquastic by those rotations if conditions \ref{diag}, \ref{perm}, and \ref{negative} are satisfied. Therefore it suffices to show that $sol$ is a set of $SO(3)$ rotations satisfying those conditions. The construction of $sol$ is primarily described in steps 12 and 13. 

Conditions \ref{diag} and \ref{perm} are only determined by the assignment of permutations  $\{\Pi_u\}$. The assignment of permutations is determined by a compatible set of permutation assignments $A= \{ \Sigma_v\}$, so for every low-rank edge, conditions \ref{diag} and \ref{perm} will be satisfied by definition.  For every \highrank\ connected component $\Gamma$, associated with a vertex $v$ in the \quotientGraph , every vertex in $\Gamma$ is assigned the same permutation $\Pi \in \Sigma_v$. Thus by the argument given in lemma \ref{rigid}, condition \ref{diag} is satisfied for all edges in $\Gamma$. Furthermore, since $\Sigma_v$ only contains permutations compatible with the admissible permutations of the \highrank\ connected component, \ref{perm} must also be satisfied. Thus all edges satisfy conditions \ref{perm} and \ref{diag}.

For every connected component of the partitioned \basegraph\, and ever vertex $u$ in the connected component, the sign assignments $R_u = \textrm{diag}(\delta_u, 1, \delta_u \det(\Pi_u))$ are determined by the solution $\{ \delta_u \}$ to an Ising problem given by definition \ref{isingSign}. The solution must satisfy $\delta_u \left(-{\rm sign}([\Pi_u \beta^{uv} \Pi_v^T]_{11}) \right) \delta_v = 1$. Therefore  $[R_u \Pi_u \beta^{uv} \Pi_v^T R_v^T]_{11} \leq 0$ for every connected component in the partitioned \basegraph\. 
The only edges which do not belong to a connected component in the partitioned \basegraph\ are those edges for which $[R_u \Pi_u \beta^{uv} \Pi_v^T R_v^T]_{11}=[\Pi_u \beta^{uv} \Pi_v^T]_{11}=0$. Therefore for all edges, condition \ref{negative} is satisfied.

Finally, every $\sPi_u$ is an $SO(3)$ transformation, since $\det(\sPi_u) = \det(R_u) \det(\Pi_u) = \delta_u^2 \det(\Pi_u)^2 =+1$.

\end{proof}

\section{Discussion}

There remain many open questions with regards to stoquasticity. It would be interesting to capture the two-qubit (or two-qudit) stoquastic space in terms of inequalities involving invariants. It would be desirable if the problem of deciding whether a Hamiltonian is computationally stoquastic is easier when allowing more transformations than just local basis changes. 

Many of the proofs in this paper are tedious and involved since the reasoning is different for different subcase Hamiltonians. We did not see a way of simplifying these arguments. It is an interesting question whether one could supply such proofs using a proof assistant, automatically ensuring their validity.

The presence of single-qubit terms causes our algorithm for deciding the stoquasticity of the XYZ Heisenberg model to fail. Furthermore, the existence of single-qubit terms makes the class of single-qubit Clifford rotations strictly weaker than that of single-qubit rotations. Nevertheless, it seems likely that a generalization of the algorithm described here could be constructed to decide stoquasticity under single-qubit Clifford rotations of an XYZ model with single qubit terms included. It remains an open question whether it is computationally hard to decide whether such a Hamiltonian is stoquastic by local basis changes beyond single-qubit Clifford rotations.

We expect that the problem of deciding whether a Hamiltonian is stoquastic by local basis changes is easy for $n$-qubit Hamiltonians with an underlying line or tree interaction graph $G$ by employing a combination of an efficient dynamic programming strategy with the parsimonious strategy (i.e. discretize the set of local basis changes, show how the optimization of a partial problem giving a basis change at the boundary can be used to solve a next-larger partial problem etc.).

An intriguing idea is that results showing hardness of deciding whether a Hamiltonian is stoquastic could be used for quantum computer verification. A general adiabatic computation with frustration-free Hamiltonians requires a quantum computer to run it, but let's assume that Alice, the verifier, picks a stoquastic frustration-free Hamiltonian, but then hides this stoquasticity by local basis changes or, say, the application of a constant depth circuit. Bob with a quantum computer can run the adiabatic computation and provide Alice with samples from its output which she can efficiently verify to be statistically correct using the algorithm in \cite{Bravyi2008}. On the other hand, a QC-pretender Charlie may have to either find the constant-depth hiding circuit or classically simulate the quantum adiabatic algorithm, and either task could be too daunting. A similar idea has been suggested in \cite{Marvian2018}.

\section{Funding}
We acknowledge support through the EU via the ERC GRANT EQEC No. 682726. This research was supported in part by Perimeter Institute for Theoretical Physics. Research at Perimeter Institute is supported by the Government of Canada through Industry Canada and by the Province of Ontario through the Ministry of Economic Development $\&$ Innovation.

\section{Acknowledgements}
BMT would like to thank David DiVincenzo for interesting discussions. BMT and JDK would also like to thank Sergey Bravyi as well as Christophe Vuillot for helpful discussions. 

\bibliographystyle{utphys}
\bibliography{bibliography-stoq-rev}

\providecommand{\href}[2]{#2}\begingroup\raggedright\begin{thebibliography}{10}

\bibitem{Bravyi2007}
S.~Bravyi, D.~P. DiVincenzo, R.~I. Oliveira, and B.~M. Terhal, ``{The
  Complexity of Stoquastic Local Hamiltonian Problems},'' 
\href{ https://doi.org/10.26421/QIC8.5}{{\em Quantum
  Information and Computation} {\bfseries 8} no.~5, (2008) 0361--0385} ,  \href{http://arxiv.org/abs/0606140}{{\ttfamily arXiv:0606140 [quant-ph]}}.

\bibitem{CM:projection}
N.~J. Cerf and O.~C. Martin, ``{Projection Monte Carlo methods : an algorithmic
  analysis},'' \href{http://dx.doi.org/10.1142/S0129183195000587}{{\em
  International Journal of Modern Physics C} {\bfseries 6} no.~5, (1995)
  693--723}.

\bibitem{Sorella2000}
S.~Sorella and L.~Capriotti, ``{Green function Monte Carlo with stochastic
  reconfiguration: An effective remedy for the sign problem},''
  \href{http://dx.doi.org/10.1103/PhysRevB.61.2599}{{\em Physical Review B - Condensed Matter and Materials Physics} {\bfseries 61} no.~4, (2000)
  2599--2612}, \href{http://arxiv.org/abs/9902211}{{\ttfamily arXiv:9902211
  [cond-mat]}}.

\bibitem{Bravyi:guiding}
S.~Bravyi, ``{Monte Carlo Simulation of Stoquastic Hamiltonians},'' \href{http://dx.doi.org/10.26421/QIC15.13-14}{\em
  Quantum Information and Computation} {\bfseries 15} no.~13/14, (2015)
  1122--1140, \href{http://arxiv.org/abs/1402.2295}{{\ttfamily
  arXiv:1402.2295}}.

\bibitem{wessel:lecture}
S.~Wessel, ``{Monte Carlo Simulations of Quantum Spin Models Institute for
  Theoretical Solid State Physics},'' in {\em Autumn School on Correlated
  Electrons}.
\newblock 2013.
\newblock \url{https://www.cond-mat.de/events/correl13/manuscripts/}.

\bibitem{thesis:crosson}
E.~Crosson, {\em {Classical and Quantum Computation in Ground States and
  Beyond}}.
\newblock PhD thesis, University of Washington, 2015.
\newblock \url{http://hdl.handle.net/1773/34128}.

\bibitem{Bravyi2016}
S.~Bravyi and D.~Gosset, ``{Polynomial-Time Classical Simulation of Quantum
  Ferromagnets},'' \href{http://dx.doi.org/10.1103/PhysRevLett.119.100503}{{\em
  Physical Review Letters} {\bfseries 119} no.~10, (2017) },
  \href{http://arxiv.org/abs/1612.05602}{{\ttfamily arXiv:1612.05602}}.

\bibitem{AL:adiabatic}
T.~{Albash} and D.~A. {Lidar}, ``{Adiabatic quantum computation},''
  \href{http://dx.doi.org/10.1103/RevModPhys.90.015002}{{\em Reviews of Modern
  Physics} {\bfseries 90} no.~1, (2018) 015002},
  \href{http://arxiv.org/abs/1611.04471}{{\ttfamily arXiv:1611.04471
  [quant-ph]}}.

\bibitem{Bravyi2008}
S.~Bravyi and B.~Terhal, ``{Complexity of stoquastic frustration-free
  Hamiltonians},'' \href{http://dx.doi.org/10.1137/08072689X}{{\em SIAM J.
  Comput.} {\bfseries 39} no.~4, (2009) 1642},
  \href{http://arxiv.org/abs/0806.1746}{{\ttfamily arXiv:0806.1746}}.

\bibitem{Hastings}
M.~B. Hastings and M.~H. Freedman, ``{Obstructions To Classically Simulating
  The Quantum Adiabatic Algorithm},''
  \href{https://doi.org/10.26421/QIC13.11-12}{{\em Quantum
  Information and Computation} {\bfseries 13} no.11/12, (2013) 1038-1076}
  \href{http://arxiv.org/abs/1302.5733}{{\ttfamily arXiv:1302.5733}}.

\bibitem{bringewatt:diff}
J.~{Bringewatt}, W.~{Dorland}, S.~P. {Jordan}, and A.~{Mink}, ``{Diffusion
  Monte Carlo approach versus adiabatic computation for local Hamiltonians},''
  \href{http://dx.doi.org/10.1103/PhysRevA.97.022323}{{\em Phys. Rev. A}
  {\bfseries 97} no.~2, (Feb., 2018) 022323},
  \href{http://arxiv.org/abs/1709.03971}{{\ttfamily arXiv:1709.03971
  [quant-ph]}}.

\bibitem{kafri:nonstoq}
D.~Kafri, C.~Quintana, Y.~Chen, A.~Shabani, J.~M. Martinis, and H.~Neven,
  ``{Tunable inductive coupling of superconducting qubits in the strongly
  nonlinear regime},'' \href{http://dx.doi.org/10.1103/PhysRevA.95.052333}{{\em
  Physical Review A} {\bfseries 95} no.~5, (May, 2017) 052333},
  \href{http://arxiv.org/abs/1606.08382}{{\ttfamily arXiv:1606.08382}}.

\bibitem{hormozi:anneal}
L.~{Hormozi}, E.~W. {Brown}, G.~{Carleo}, and M.~{Troyer}, ``{Nonstoquastic
  Hamiltonians and quantum annealing of an Ising spin glass},''
  \href{http://dx.doi.org/10.1103/PhysRevB.95.184416}{{\em Phys. Rev. B}
  {\bfseries 95} no.~18, (May, 2017) 184416},
  \href{http://arxiv.org/abs/1609.06558}{{\ttfamily arXiv:1609.06558
  [quant-ph]}}.

\bibitem{samach:talk}
G.~Samach, ``{Tunable XX-Coupling Between High Coherence Flux Qubits},'' in
  {\em APS March Meeting 2018}.
\newblock 2018.
\newblock \url{https://meetings.aps.org/Meeting/MAR18/Session/L33.13}.

\bibitem{IKS:sign}
V.~I. Iglovikov, E.~Khatami, and R.~T. Scalettar, ``Geometry dependence of the
  sign problem in quantum {M}onte {C}arlo simulations,''
  \href{http://dx.doi.org/10.1103/PhysRevB.92.045110}{{\em Phys. Rev. B}
  {\bfseries 92} (Jul, 2015) 045110}.
  \url{https://link.aps.org/doi/10.1103/PhysRevB.92.045110}.

\bibitem{WZ:sign}
C.~Wu and S.-C. Zhang, ``Sufficient condition for absence of the sign problem
  in the fermionic quantum monte carlo algorithm,'' \href{http://dx.doi.org/10.1103/PhysRevB.71.155115}{{\em Phys. Rev. B}
  {\bfseries 71} (Apr, 2005) 155115}.

\bibitem{LY:review}
Z.-X. {Li} and H.~{Yao}, ``{Sign-Problem-Free Fermionic Quantum {M}onte
  {C}arlo: Developments and Applications},'' {\em arXiv e-prints} (May, 2018)
  arXiv:1805.08219, \href{http://arxiv.org/abs/1805.08219}{{\ttfamily
  arXiv:1805.08219 [cond-mat.str-el]}}.

\bibitem{BF:MPrule}
R.~F. Bishop and D.~J.~J. Farnell, ``{M}arshall-{P}eierls sign rules, the
  quantum monte carlo method, and frustration,'' \href{http://dx.doi.org/10.1142/9789812792754_0052}{{\em International Journal of
  Modern Physics B} {\bfseries 15} no.~10n11, (2001) 1736--1739}.



\bibitem{Marvian2018}
M.~Marvian, D.~A. Lidar, and I.~Hen, ``{On the Computational Complexity of
  Curing non-stoquastic Hamiltonians},'' \href{http://dx.doi.org/10.1038/s41467-019-09501-6}{{\em Nature Communications} {\bfseries 10} no. 1, (2019) 1571} ,
  \href{http://arxiv.org/abs/1802.03408}{{\ttfamily arXiv:1802.03408}}.

\bibitem{terhal:talk}
B.~M. Terhal, ``{The Power and Use of Stoquastic Hamiltonians},'' in {\em
  Adiabatic Quantum Computing Conference}.
\newblock 2017.
\newblock \url{https://www.youtube.com/watch?v=4dK30QExF4M}.

\bibitem{CMP:universal}
T.~Cubitt, A.~Montanaro, and S.~Piddock, ``{Universal Quantum Hamiltonians},''
  \href{https://doi.org/10.1073/pnas.1804949115 }{{\em National Academy of Sciences} {\bfseries 115} no. 38 (2018) 9497-9502 },
  \href{http://arxiv.org/abs/1701.05182}{{\ttfamily arXiv:1701.05182}}.

\bibitem{BH:ising}
S.~{Bravyi} and M.~{Hastings}, ``{On complexity of the quantum Ising model},'' \href{https://doi.org/10.1007/s00220-016-2787-4}{
  {\em Communications in Mathematical Physics} {\bfseries 349} no. 1 (2017) 1-45} ,
  \href{http://arxiv.org/abs/1410.0703}{{\ttfamily arXiv:1410.0703
  [quant-ph]}}.

\bibitem{GS:clifford}
D.~{Grier} and L.~{Schaeffer}, ``{The Classification of Stabilizer Operations
  over Qubits},'' {\em ArXiv e-prints} (Mar., 2016) ,
  \href{http://arxiv.org/abs/1603.03999}{{\ttfamily arXiv:1603.03999
  [quant-ph]}}.

\bibitem{Makhlin2000}
Y.~Makhlin, ``{Nonlocal properties of two-qubit gates and mixed states and
  optimization of quantum computations},''
  \href{http://dx.doi.org/10.1023/A:1022144002391}{{\em Quantum Information
  Processing} {\bfseries 1} no.~4, (2002) 243--252},
  \href{http://arxiv.org/abs/0002045}{{\ttfamily arXiv:0002045 [quant-ph]}}.

\bibitem{Linden1999}
N.~Linden, S.~Popescu, and A.~Sudbery, ``{Nonlocal Parameters for Multiparticle
  Density Matrices},'' \href{http://dx.doi.org/10.1103/PhysRevLett.83.243}{{\em
  Physical Review Letters} {\bfseries 83} no.~2, (1999) 243--247},
  \href{http://arxiv.org/abs/9801076}{{\ttfamily arXiv:9801076 [quant-ph]}}.

\bibitem{Grassl:inv}
M.~{Grassl}, M.~{R{\"o}tteler}, and T.~{Beth}, ``{Computing local invariants of
  quantum-bit systems},''
  \href{http://dx.doi.org/10.1103/PhysRevA.58.1833}{{\em Phys. Rev. A}
  {\bfseries 58} (Sept., 1998) 1833--1839},
  \href{http://arxiv.org/abs/quant-ph/9712040}{{\ttfamily quant-ph/9712040}}.

\bibitem{Bertlmann2008a}
R.~A. Bertlmann and P.~Krammer, ``{Bloch vectors for qudits},''
  \href{http://dx.doi.org/10.1088/1751-8113/41/23/235303}{{\em Journal of
  Physics A: Mathematical and Theoretical} {\bfseries 41} no.~23, (2008) },
  \href{http://arxiv.org/abs/0806.1174}{{\ttfamily arXiv:0806.1174}}.

\bibitem{Gonzalez1985}
T.~F. Gonzalez, ``{Clustering to minimize the maximum intercluster distance},''
  \href{http://dx.doi.org/10.1016/0304-3975(85)90224-5}{{\em Theoretical
  Computer Science} {\bfseries 38} (1985) 293--306}.

\end{thebibliography}\endgroup

\appendix

\section{Symmetric Z-Matrix Cone in $\mathbb{C}^{d^2}$} \label{AppendixNecessaryAndSufficientConditions}

It is understood that the set of symmetric Z-matrices forms a polyhedral cone. For convenience and concreteness we present the structure of this cone for operators acting on $\mathbb{C}^{d^2}$ in terms of generalizations of the Pauli matrices.

Any Hermitian operator acting on $\mathbb{C}^d \otimes \mathbb{C}^d$ with basis elements $\{\ket{i} \;\vert \; i\in {0,..d-1}\}$, can be written as :
$$ H=\sum_{x,y \in {D,A,S}}\sum_{i,j,m,n=0}^{d-1} a_{ij;mn}^{xy} \lambda^x_{ij} \otimes \lambda^y_{mn},$$
with $a_{ij;mn}^{xy}\in \mathbb{R}$. Here the $\lambda^{D,A,S}_{ij}$ matrices with labels $S$ (symmetric), $A$ (anti-symmetric) and $D$ (diagonal) are a modified set of generalized Gell-Mann matrices \cite{Bertlmann2008a} given by:

\begin{align}
\lambda^S_{ij} &= \left\lbrace \begin{matrix}
\sqrt{d/2}\left(\ket{i}\bra{j}+ \ket{j} \bra{i}\right) &\; : \; i< j \\
0 & \; : \; \text{otherwise}
\end{matrix}\right.\\
\lambda^A_{ij} &=\left\lbrace \begin{matrix}
-i\sqrt{d/2} \left(\ket{i}\bra{j} - \ket{j}\bra{i}\right) &\; : \; i< j \\
0 & \; : \; \text{otherwise}
\end{matrix}\right.\\
\lambda^D_{ij} &= \left\lbrace \begin{matrix}
\frac{\sqrt{d}}{\sqrt{d}-1} e_0 + \sqrt{d} e_i &\; : \;i=j>0 \\
\mathbb{I} & \; : \; i=j=0\\
0 &\;: \; \text{otherwise}
\end{matrix} \right.
\end{align}
with $e_i =\ket{i}\bra{i}-\frac{1}{d} \mathbb{I}$. Note that ${\rm Tr}[\lambda_{ij}^x \lambda_{mn}^y] =\delta_{im}\delta_{jn}\delta_{xy} d$, so that the $\lambda$ matrices can be thought of as spanning vectors of the Hilbert-Schmidt vector space. We can thus employ the vector inner product of two bipartite matrices $A$ and $B$ acting on $\mathbb{C}^d \otimes \mathbb{C}^d$ as: $\langle A, B \rangle ={\rm Tr}[A^{\dagger}B]/d^2$.

 The necessary and sufficient conditions for $H$ to be a symmetric $Z$-matrix (see Definition \ref{def:Z}) are (1) $\forall i,j,m,n$
$\bra{im}H\ket{jn} \in \mathbb{R}$ and (2) $\bra{im}H\ket{jn} \leq 0$ when either $i \neq j$ or $m \neq n$ (or both).

 The $\lambda$-basis is convenient because of the following properties. The diagonal part of $H$ only has support on vectors of the form $\lambda^D \otimes \lambda^D$ while the off-diagonal part has no contribution from $\lambda^D \otimes \lambda^D$. Furthermore, the imaginary part of $H$ only has support on vectors of the form $\lambda^A \otimes \lambda^S$, $\lambda^A \otimes \lambda^D$, $\lambda^S \otimes \lambda^A$ and $\lambda^D \otimes \lambda^A$, while both the real and the diagonal part of $H$ have no contribution from these vectors. Thus the vector $H$ has support on three orthogonal subspaces, the diagonal subspace $D$ spanned by vectors $\lambda^D \otimes \lambda^D$, the imaginary subspace $I$, spanned by vectors $\lambda^A \otimes \lambda^S$, $\lambda^A \otimes \lambda^D$, $\lambda^S \otimes \lambda^A$ and $\lambda^D \otimes \lambda^A$, and the off-diagonal real subspace $R$ spanned by the remaining vectors.

We may write $H = H_D + H_R + H_I$ where $H_D,H_R,H_I$ are the vectors in these three subspaces.

The first (``realness'') condition thus corresponds to demanding that $H_I=0$. By the linear independence of the $\lambda$-basis it implies that
\begin{equation}
\forall i,j,m,n,\; a^{AS}_{ij;mn}=a^{SA}_{ij;mn}=a^{AD}_{ij;mn}=a^{DA}_{ij;mn}=0.
\end{equation}

The second (``negativity'') condition then says that for all $i \neq j$ or $m \neq n$ (or both):
$$ \bra{im}H\ket{jn}= {\rm Tr}[\ket{jn}\bra{im}H]= d^2 \langle \ket{jn}\bra{im} , H_R \rangle \leq 0.$$

This condition thus specifies some of the facets of the symmetric $Z$-matrix cone. When a vector $\ket{im}\bra{jn}$ is interpreted as the normal vector of an oriented hyperplane in Hilbert-Schmidt space, the negativity condition can be thought of as the requirement that $H_R$ lies within the union of half spaces defined by these oriented hyperplanes.

 To write down the inequalities describing these facets, we can use that
$$
\ket{j}\bra{i} = \left\lbrace  \begin{matrix}
\frac{1}{2}(\lambda^S_{ij}-i \lambda^A_{ij}) &\; : \; i<j, \\
\frac{1}{2}(\lambda^S_{ji}+i \lambda^A_{ji}) &\; : \; i>j, \\
e_i+\frac{1}{d}\mathbb{I} &\; : \; i=j,
\end{matrix}\right.
$$
with
\begin{equation}
e_0 = \frac{1}{d} \sum_{m=1}^{d-1} \lambda_m^D, \;e_{i\neq 0} = \frac{1}{\sqrt{d}} \lambda_i^D- \frac{1}{d(\sqrt{d}-1)} \sum_{m=1}^{d-1} \lambda_m^D. \nonumber
\end{equation}

 One can then write down three types of inequalities:
\begin{itemize}
\item case 1: $i \neq j$ and $m\neq n$
\begin{equation} \label{generalStoqCond1}
\begin{split}
&\langle \ket{jn}\bra{im} , H_R \rangle  \\ &=\left\lbrace \begin{matrix}
\frac{1}{4} \left(a^{SS}_{ij;mn} - a^{AA}_{ij;mn} \right)\leq 0 \;:&\; \begin{matrix}(i<j,\; m<n )  \\ (i>j,\; m>n) \end{matrix}\\
\\
\frac{1}{4} \left( a^{SS}_{ij;mn} + a^{AA}_{ij;mn} \right)\leq 0 \;:&\; \begin{matrix} (i>j,\; m<n )\\ (i<j,\; m>n) \end{matrix}\\
\end{matrix}\right.
\end{split}
\end{equation}
\item case 2: $i=j$ and $m \neq n$
\begin{equation}\label{generalStoqCond2}
\begin{split}
&\langle \ket{in}\bra{im} , H_R \rangle =\frac{1}{2} \left\langle\left(e_i+\mathbb{I}/d\right)\otimes \lambda_{mn}^S, H_R \right\rangle\\&= \frac{1}{2}\left(\langle e_i \otimes \lambda_{mn}^S, H_R \rangle+a_{00;mn}^{DS}/d\right) \leq 0
\end{split}
\end{equation}
\item case 3: $i \neq j$ and $m=n$
\begin{equation} \label{generalStoqCond3}
\begin{split}
& \left\langle \ket{jm}\bra{im} , H_R \right\rangle=\frac{1}{2} \left\langle\lambda_{ij}^S\otimes\left(e_m+\mathbb{I}/d\right) , H_R \right\rangle\\ & = \frac{1}{2}\left(\langle \lambda_{ij}^S \otimes e_m, H_R \rangle+a_{ij;00}^{SD}/d\right) \leq 0
\end{split}
\end{equation}
\end{itemize}

There is an interesting geometric fact here. Consider for instance Inequality \eqref{generalStoqCond2}. For a fixed $n$ and $m\neq n$ Inequality \eqref{generalStoqCond2} is only concerned with the elements of the vector $H_R$ with support on basis elements of the form $\lambda^D \otimes \lambda_{mn}^S$, which do not appear in any other inequalities. So for a fixed $n$ and $m\neq n$ we may think of the set of inequalities as living in a $d$-dimensional subspace of Hilbert-Schmidt space orthogonal to all other sets of inequalities. The vectors $(\ket{i} \bra{i} \otimes \ket{n}\bra{m})_R=\left(e_i+\frac{\mathbb{I}}{d}\right)\otimes \lambda_{mn}^S$ are the normal vectors of the supporting hyperplanes of our convex set in this space, and the vectors $\{e_m\}$ correspond to the extremal points of a regular simplex. Fixing  a value for $\langle\frac{\mathbb{I}}{d}\otimes\lambda_{mn}^S, H_R \rangle=\frac{a_{ij;00}^{SD}}{d}$ we can take a $d-1$-dimensional slice of our convex set and see that the vectors  $e_i \otimes \lambda_{mn}^S$ form the normals of a new supporting hyperplane whose distances from the origin of the slice is proportional to our value of  $\frac{a_{ij;00}^{SD}}{d}$. Thus inequalities \eqref{generalStoqCond2} and \eqref{generalStoqCond3} give the polyhedral cone a simplicial structure.

This is illustrated in figure \ref{simplicialCone} for $d=2,3$.

\begin{figure}[h]
\center
\includegraphics[width=.43 \linewidth]{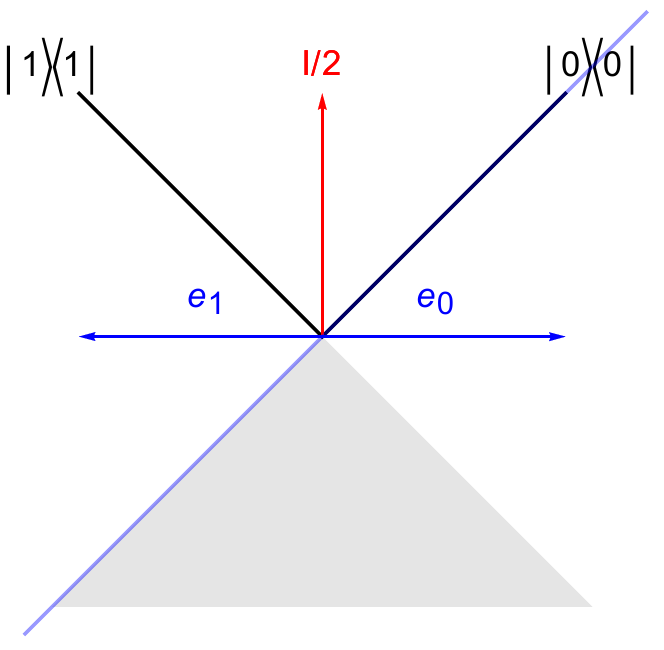}
\includegraphics[width=.55\linewidth]{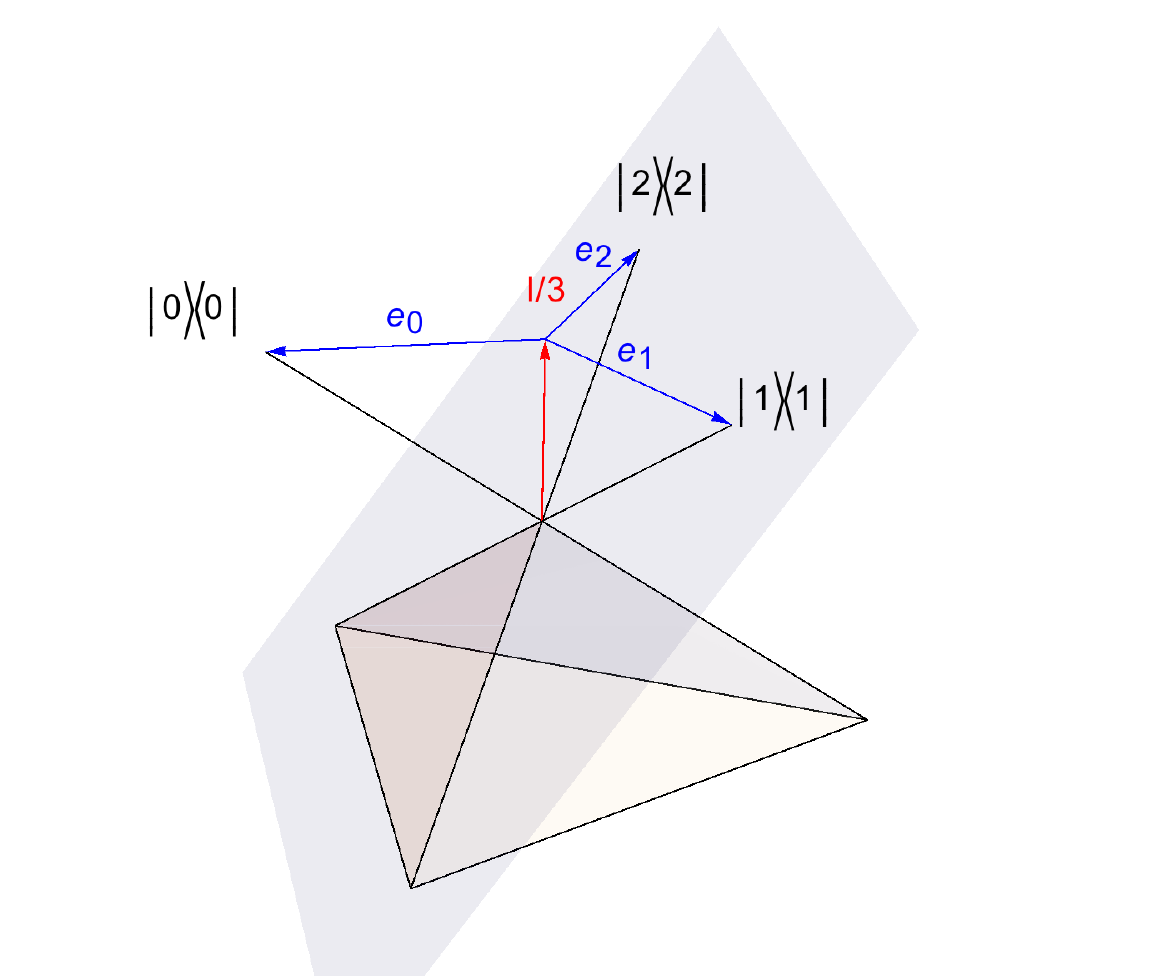}
\caption{An illustration of the simplicial polyhedral cones generated by the inequalities \eqref{generalStoqCond2} or \eqref{generalStoqCond3} for $d=2$ and $d=3$. For example in the case of Inequality \eqref{generalStoqCond2}, fixing $n$ and $m\neq n$ gives solutions lying inside a cone bounded by the set of hyperplanes defined by the normal vectors $(\ket{i} \bra{i} \otimes \ket{n}\bra{m})_R$ (labelled $\ket{i} \bra{i}$ in the figures). See for example in the $d=3$ figure the blue plane defined by the normal vector $(\ket{0} \bra{0} \otimes \ket{n}\bra{m})_R$. Note that the cone lies in the lower half of the space, for example the region where $\langle \frac{\mathbb{I}}{d}\otimes\lambda^S_{mn}, H_R \rangle \leq 0$ in the case of Inequality \eqref{generalStoqCond2}.} \label{simplicialCone}
\end{figure}

In the case of two qubits $\lambda^S_{01}=X$, $\lambda^A_{01} =Y$, $\lambda^D_{01}=Z$, $e_0 =Z/2$ and $e_1=-Z/2$. For the realness condition we require that $a_{IY}=a_{YI}=a_{YZ}=a_{ZY}=a_{YX}=a_{XY} =0$. For the negativity condition we retrieve the inequalities: 
\begin{align*}
a_{XX}&\leq  a_{YY}&  a_{XX}  &\leq -a_{YY} \\
a_{IX} &\leq -a_{ZX}&  a_{IX} &\leq a_{ZX}  \\
a_{XI} &\leq -a_{XZ}& a_{XI} &\leq a_{XZ}.
\end{align*}

Together these (in)equalities are given in condensed form in Proposition \ref{stoquasticCondition}.

\section{Proof of Theorem \ref{theo:real}} \label{AppendixLocalRealnessCondition}
 
Recall the definitions of
$$\Gamma_L = \left\lbrace S, \beta \beta^T S,  [\beta \beta^T]^2 S, \beta P , \beta \beta^T \beta P \right\rbrace,$$
and 
$$\Gamma_R = \left\lbrace P, \beta^T\beta P , [\beta^T \beta]^2 P ,  \beta^T S  , \beta^T \beta \beta^T S \right\rbrace.$$

When the triple product invariants are zero, $\Gamma_L$ and $\Gamma_R$ each contain coplanar vectors. By Proposition \ref{span} (see below) we know that $S$ is in the span of at most two left singular vectors of $\beta$ and $P$ is in the span of at most two right singular vectors of $\beta$. 

If $S$, $ \beta \beta^T S$ and $ [\beta \beta^T]^2 S$ are linearly independent, then $\beta P$ must be in the same span as $S$. If they are not linearly independent then $S$ is a left singular vector of $\beta$, with unit vector $\hat{e}_S$. It follows that $\beta P \in \text{span} (\hat{e}_S, \hat{e}_{S_\perp})$, and $\beta \beta^T \beta P \in   \text{span} (\hat{e}_S, \hat{e}_{S_\perp})$. This implies that $\beta \beta^T \hat{e}_{S_\perp} \in \text{span}(\hat{e}_S, \hat{e}_{S_\perp})$. Since $\hat{e}_S^T \beta \beta^T \hat{e}_{S_\perp}=0$ it follows that $\beta \beta^T \hat{e}_{S_\perp} \in \text{span}(\hat{e}_{S_\perp})$ and $\hat{e}_{S_\perp}$ is a left singular vector of $\beta$. Therefore there always exists a pair of left singular vectors of $\beta$ such that $S$ and $\beta P$ are in their span. By Proposition \ref{betaP} (see below) this implies that $S$ is spanned by at most two left singular vectors of $\beta$ and $P$ is spanned by the corresponding right singular vectors. By Proposition \ref{realH1} (see below) we see that $H$ is real under local unitary rotations. $\Box$

\begin{proposition} \label{span}
If $S$, $\beta \beta^T S$ and $ [\beta \beta^T]^2 S$ are coplanar, then $\exists \hat{e}_x^L,\hat{e}_y^L \text{ s.t. } S \in \text{span}(\hat{e}_x^L,\hat{e}_y^L )$ with $\hat{e}_x^L,\hat{e}_y^L$ left singular vectors of $\beta$. Similarly for $P$.
\end{proposition}

\begin{proof}
Consider the vectors $S$, $\beta \beta^T S$ and $ [\beta \beta^T]^2 S$. Consider the expansion of $S$ into an orthonormal set of left singular vectors of $\beta$: $S=\sum_i^3 S_i \hat{e}_i^L$. If $\beta \beta^T$ is not degenerate, then $S$, $\beta \beta^T S$ and $ [\beta \beta^T]^2 S$ are coplanar only if $\exists i \text{ s.t. } S_i = 0$, which implies that $S \in \text{span}(\hat{e}_j^L,\hat{e}_k^L )$. If $\beta\beta^T$ is two fold degenerate then $\exists j,k \text{ s.t. } S_j \hat{e}_j^L + S_k \hat{e}_k^L$ is a left singular vector of $\beta \beta^T$, which implies that $S\in \text{span}(\hat{e}_i^L,\hat{e}_{jk}^L )$ with $\hat{e}_{jk}^L$ a left singular vector of $\beta$. Finally, if $\beta \beta^T$ is three fold degenerate, then $S$ is a left singular vector of $\beta$. Therefore, regardless of $\beta$, $\exists \hat{e}_x^L,\hat{e}_y^L \text{ s.t. } S \in \text{span}(\hat{e}_x^L,\hat{e}_y^L )$ with $\hat{e}_x^L,\hat{e}_y^L$ left singular vectors of $\beta$. A symmetric argument holds for $P$.
\end{proof}

\begin{proposition}\label{betaP}
If $\Gamma_L$ and $\Gamma_R$ each contain coplanar vectors, then if there exist a pair of left singular vectors $ \hat{e}_x^L,\hat{e}_y^L$ such that $\beta P, S \in \text{span}(\hat{e}_x^L,\hat{e}_y^L )$ then the right singular vectors which span $P$ correspond to the left singular vectors which span $S$. 
\end{proposition}

\begin{proof}
Let $S \in \text{span}(\hat{e}_x^L, \hat{e}_y^L)$. Assume $\beta P$ is in the same span as $S$, then $\beta^T \beta P, [\beta^T \beta]^2 P \in \text{span}(\hat{e}_x^R, \hat{e}_y^R)$. If $\beta^T \beta P$ and $[\beta^T \beta]^2 P$ are not proportional to one another then $P \in \text{span}(\beta^T\hat{e}_x^L, \beta^T\hat{e}_y^L)$. If $\beta^T \beta P$ and $[\beta^T \beta]^2 P$ are proportional to one another then they are proportional to a right singular vector of $\beta$, $\hat{e}_P^R$. If $\beta$ is full rank then $\beta^T \beta$ is invertible and $P \propto \hat{e}_P^R \in \text{span}(\hat{e}_x^R, \hat{e}_y^R)$. If $\beta$ is not full rank and $P \not\propto \hat{e}_P^R$ then $P \in \text{span}(\hat{e}_P^R, \hat{e}_0^R)$ with $\beta \hat{e}_0^R=0$. Consider the left singular vector expansion of $S$ into an orthonormal basis: $S=S_P\hat{e}_P^L + S_0\hat{e}_0^L+ S_{\perp} \hat{e}_{\perp}^L$. We require the following three vectors to be coplanar 
$$ \beta^T S\in \text{span}(\hat{e}_P^R,\hat{e}_{\perp}^R)$$
$$P\in \text{span}(\hat{e}_P^R, \hat{e}_0^R)$$
$$\beta^T \beta P \in \text{span}(\hat{e}_P^R).$$
The only way these three vectors can be coplanar is if $S_{\perp}=0$ or $\beta^T \hat{e}_{\perp}^L =0$. If $\beta$ is rank 1 and $\beta^T \hat{e}_{\perp}^L =0$ then $S_0 \hat{e}^L + S_{\perp} \hat{e}_{\perp}^L$ constitutes a left singular vector corresponding to $\hat{e}_0^R$, since $\beta \hat{e}_0^R = 0 (S_0 \hat{e}_0^L + S_{\perp} \hat{e}_{\perp}^L)$ and $\beta^T(S_0 \hat{e}^L + S_{\perp} \hat{e}_{\perp}^L)=0 \hat{e}_0^R$ and so $S$ and $P$ share a common pair of singular vectors. If $\beta$ is rank 2 then $S_{\perp}=0$ and $S$ and $P$ share common left right singular vectors $\hat{e}_P$ and $\hat{e}_0$. \end{proof}

\begin{proposition} \label{realH1}
Given a two-qubit Hamiltonian $H$ parametrized by $\beta$, $S$ and $P$.
$H$ is real under local unitary rotations iff $S$ is in the span of at most two left singular vectors of $\beta$ and $P$ is in the span of the corresponding right singular vectors.
\end{proposition}

\begin{proof}
If $H$ is real under local unitary rotations, then there exists a pair of rotations $O_L, O_R \in SO(3)$ s.t. $$
O_L \beta O_R^T =\beta'=\left(\begin{array}{ccc}
\lambda_1 &0&0 \\
0&\lambda_2 &0 \\
0&0&\lambda_3
 \end{array}\right),$$ $$O_L S= S'=\left( \begin{array}{c}
 S_1' \\
 0\\
 S_3'
 \end{array} \right),$$
$$O_R P= P'=\left( \begin{array}{c}
 P_1' \\
 0\\
 P_3' 
 \end{array} \right).$$
Expressed in the left and right basis $\{\hat{e}_i^L\}$ and $\{\hat{e}_i^R\}$ respectively. The left singular vectors of $\beta$ are $\{O_L^T\hat{e}_i^L\}$ and the corresponding right singular vectors are $\{\sign(\lambda_i)O_R^T\hat{e}_i^R\}$.

It is easy to see that $S$ and $P$ are spanned by at most two left singular vectors of $\beta$   
$$S = S_1' O_L^T \hat{e}_1^L + S_3' O_L^T \hat{e}_3^L$$
\begin{align*} 
P =  &\sign (\lambda_1) P_1' \left(\sign (\lambda_1)O_R^T \vec{e}_1^R \right) + \\ &\sign (\lambda_1) P_3' \left(\sign (\lambda_3)O_R^T \hat{e}_3^R \right)
\end{align*}

That concludes the proof for one direction of the biconditional.

Suppose $S$ and $P$ are spanned by two or fewer left and right (respectively) singular vectors of $\beta$ which share the same singular values. Then $S$ and $P$ are expressible as 
$$S = S_1 \hat{e}_1^L + S_3 \hat{e}_3^L$$ 
$$P = P_1 \hat{e}_1^R + P_3 \hat{e}_3^R$$ 
with
$$\beta^T \hat{e}_i^L = \vert \lambda_i \vert \hat{e}_i^R$$
and
$$ \beta \hat{e}_i^R = \vert \lambda_i \vert \hat{e}_i^L$$

There exists a pair of rotations $O_L, O_R \in SO(3)$ s.t. 

$$
O_L \beta O_R^T =\beta'=\left(\begin{array}{ccc}
\lambda_1 &0&0 \\
0&\lambda_2 &0 \\
0&0&\lambda_3
 \end{array}\right)$$
with
$$S' = O_LS = \left(\begin{array}{c}S_1\\0\\S_3\end{array} \right)$$
and
$$P' = O_RP=\left( \begin{array}{c}\sign(\lambda_1)P_1\\0\\\sign(\lambda_3)P_3\end{array} \right)$$
which corresponds to a real Hamiltonian. Thus H is real under local unitary rotations. 
\end{proof}

\section{Polynomial Inequalities for Stoquasticity of a 2-qubit Hamiltonian }\label{inequalityAppendix}

Given the standard form described in section \ref{sec:ineq}, rotated by $SO(3)$ rotations given by equation \ref{eq:sfRot}, the rotated two-qubit Hamiltonian is described by:

\begin{equation} \label{rotBeta} 
\beta'=\left( \begin{array}{ccc}
a_{XX}'& 0& a_{XZ}'\\
0  & a_{YY}&0\\
a_{ZX}'& 0&a_{ZZ}' 
\end{array} \right)\end{equation}
with
$$
\begin{array}{c}
a_{ZZ}'=\gamma_L \gamma_R \left(\cos(\theta_L) \cos(\theta_R)+ a_{XX} \sin(\theta_L) \sin(\theta_R) \right) \\
a_{ZX}' =\gamma_L (\cos(\theta_L) \sin(\theta_R)-a_{XX} \cos(\theta_R) \sin(\theta_L))  \\
a_{XZ}'=\gamma_R \left(\cos(\theta_R) \sin(\theta_L) -a_{XX}  \cos(\theta_L) \sin(\theta_R) \right) \\
a_{XX}'= \sin(\theta_L) \sin(\theta_R)+a_{XX} \cos(\theta_L) \cos(\theta_R) \\
a_{YY}'=a_{YY} \gamma_R \gamma_L
\end{array}
$$

 and
\begin{equation} \label{rotS} 
S'= \left( \begin{array}{c}
a_{XI} \cos(\theta_L) + a_{ZI} \sin(\theta_L) \\
0\\
a_{ZI} \gamma_L \cos(\theta_L) - a_{XI} \gamma_L \sin(\theta_L) 
\end{array} \right) \end{equation}
\begin{equation} \label{rotP}
P'= \left( \begin{array}{c}
a_{IX} \cos(\theta_R) + a_{IZ} \sin(\theta_R) \\
0\\
a_{IZ} \gamma_R \cos(\theta_R) - a_{IX} \gamma_R \sin(\theta_R) 
\end{array} \right)\end{equation}

Three inequalities must be satisfied in order for $H$ to be a \zmatrix{}:

\begin{multline} \label{ineq-set1}
\sin(\theta_L) \sin(\theta_R) +a_{XX} \cos(\theta_L) \cos (\theta_R) \leq - \vert a_{YY} \vert 
\end{multline}
\begin{multline}\label{ineq-set2}
a_{IX} \cos(\theta_R) + a_{IZ} \sin(\theta_R) \leq \\- \vert \cos(\theta_L) \sin(\theta_R) -a_{XX} \cos(\theta_R) \sin(\theta_L)\vert
\end{multline} 
\begin{multline}\label{ineq-set3}
a_{XI} \cos(\theta_L) + a_{ZI} \sin(\theta_L) \leq \\ - \vert \cos(\theta_R) \sin(\theta_L) -a_{XX} \cos(\theta_L) \sin(\theta_R) \vert 
\end{multline}

Note that the parameters $\gamma_1$ and $\gamma_2$ have fallen out of these inequalities. 

If the above inequalities can be satisfied, then the Hamiltonian is stoquastic. If they cannot be satisfied then the Hamiltonian is not stoquastic, barring a single case as follows. If either $\{ a_{XI} , a_{IX} \}=\{0,0\}$ or $\{ a_{ZI} , a_{IZ}\}=\{0,0\}$, then an additional test of the same form must be applied after first performing a permutation such that the pair of zero coefficients are now in the Y position. For example, if $\{a_{XI}, a_{IX}\}=\{0,0\}$ then one must also perform the above test on the diagonal form:

$$\beta =\left( \begin{array}{ccc}
a_{YY} & 0 & 0 \\
0 &a_{XX} & 0 \\
0 & 0 & a_{ZZ}
\end{array} \right)$$

$$\begin{array}{cc}

S=\left( \begin{array}{c}0\\
0\\ a_{ZI}
 \end{array}\right)
&
P=\left(\begin{array}{c}
0\\
0\\
 a_{IZ} \end{array}\right) 
\end{array}$$

The inequalities \eqref{ineq-set1}, \eqref{ineq-set2} and \eqref{ineq-set3} can be broken up into six inequalities in order to ignore the absolute values:

\begin{itemize}
\item[(1a)] $\sin(\theta_L) \sin(\theta_R) +a_{XX} \cos(\theta_L) \cos (\theta_R) \leq 0 $
\item[(1b)] $(\sin(\theta_L) \sin(\theta_R) +a_{XX} \cos(\theta_L) \cos (\theta_R))^2 \geq  a_{YY} ^2$
\item[(2a)] $a_{IX} \cos(\theta_R) + a_{IZ} sin(\theta_R) \leq 0$
\item[(2b)] $(a_{IX} \cos(\theta_R) + a_{IZ} sin(\theta_R))^2 \geq ( \cos(\theta_L) \sin(\theta_R) -a_{XX} \cos(\theta_R) \sin(\theta_L))^2$
\item[(3a)] $a_{XI} \cos(\theta_L) + a_{ZI} \sin(\theta_L) \leq 0$
\item[(3b)] $(a_{XI} \cos(\theta_L) + a_{ZI} \sin(\theta_L) )^2 \geq (\cos(\theta_R) \sin(\theta_L) -a_{XX} \cos(\theta_L) \sin(\theta_R) )^2 $
\end{itemize}

We wish to map these trigonometric inequalities into polynomial inequalities of two variables by defining $x_1=\tan(\theta_L)$ and $x_2 = \tan(\theta_R)$. However the mapping will depend on the value of $\cos(\theta_1)$ and $\cos(\theta_2)$, and so breaks up into several cases.

\noindent \textbf{Case 1: $\cos(\theta_L), \cos(\theta_R) \neq 0$}

\noindent Let $\delta_L = \text{Sign}(\cos(\theta_L))$ and  $\delta_R = \text{Sign}(\cos(\theta_R))$ then

\begin{itemize}
\item[(1a)] $\delta_L \delta_R (x_1 x_2 +a_{XX}  )\leq 0 $
\item[(1b)] $(x_1 x_2 +a_{XX})^2 \geq a_{YY}^2 (1+x_1^2)(1+x_2^2)$
\item[(2a)] $\delta_R (a_{IX} + a_{IZ} x_2) \leq 0$
\item[(2b)] $(a_{IX} +a_{IZ} x_2)^2(1+x_1^2) \geq (x_2 - a_{XX}x_1)^2 $
\item[(3a)] $ \delta_L (a_{XI} +a_{ZI} x_1) \leq 0$
\item[(3b)] $(a_{XI} +a_{ZI} x_1)^2 (1+x_2^2) \geq (x_1 - a_{XX} x_2)^2 $
\end{itemize}

\noindent \textbf{Case 2: $\cos(\theta_L)=0, \cos(\theta_R) \neq 0$}

\noindent Let $\delta_L = \text{Sign}(\sin(\theta_L))$ and  $\delta_R = \text{Sign}(\cos(\theta_R))$ then
\begin{itemize}
\item[(1a)] $\delta_L \delta_R x_2 \leq 0$
\item[(1b)] $x_2^2/(1+x_2^2) \geq a_{YY}^2$
\item[(2a)] $\delta_R (a_{IX} + a_{IZ} x_2) \leq 0$
\item[(2b)] $(a_{IX} +a_{IZ} x_2)^2 \geq a_{XX}^2 $
\item[(3a)] $\delta_L a_{ZI} \leq 0$
\item[(3b)] $(1+x_2^2)a_{ZI}^2 \geq 1 $
\end{itemize}

\noindent \textbf{Case 3: $\cos(\theta_L) \neq 0, \cos(\theta_R) = 0$}

\noindent Let $\delta_L = \text{Sign}(\cos(\theta_L))$ and  $\delta_R = \text{Sign}(\sin(\theta_R))$ then
\begin{itemize}
\item[(1a)] $\delta_L \delta_R x_1 \leq 0$
\item[(1b)] $x_1^2/(1+x_1^2) \geq a_{YY}^2$
\item[(2a)] $\delta_R a_{IZ} \leq 0$
\item[(2b)] $(1+x_1^2)a_{IZ}^2 \geq 1 $
\item[(3a)] $\delta_L (a_{XI} + a_{ZI} x_1) \leq 0$
\item[(3b)] $(a_{XI} +a_{ZI} x_1)^2 \geq a_{XX}^2 $

\end{itemize}

\noindent \textbf{Case 4: $\cos(\theta_L) = 0, \cos(\theta_R) = 0$}

\noindent Let $\delta_L = \text{Sign}(\sin(\theta_L))$ and  $\delta_R = \text{Sign}(\sin(\theta_R))$ then
\begin{itemize}
\item[(1a)] $\delta_L \delta_R \leq 0$
\item[(1b)] $1 \geq a_{YY}^2$
\item[(2a)] $\delta_R a_{IZ} \leq 0$
\item[(3a)] $\delta_L a_{ZI} \leq 0$
\end{itemize}

All of these inequalities are at most quadratic in any single variable $x_1$ or $x_2$, so analytic solutions to their roots can be constructed and solutions can be found using graphical methods. If there exist values of $x_1$ and $x_2$ such that for some choice of $\delta_L,\delta_R \in \{-1,1\}$ at least one of the above cases is satisfied, then $H$ is stoquastic. If not, then we may say that $H$ is not stoquastic save the exception described above.

This concludes the complete characterization of stoquasticity of a two-qubit Hamiltonian $H$.

\section{Proof of Lemma \ref{cliffDiag}}
\label{permsandrefs}

Consider a set of $3 \times 3$ matrices $\{\beta^{uv}\}$ indexed by the pair $uv$ which are diagonal in some fixed basis $\{\hat{e}_1,\hat{e}_2,\hat{e}_3\}$. Here we establish that whenever there exists a set of $SO(3)$ rotations $\{O_u\}$ mapping to a new set of diagonal matrices $\{O_u\beta^{uv}O_v^T\}$, there also exists a set of signed permutations $\{\sPi_u\}$, each with determinant 1, performing the same effective mapping: $\{O_u\beta^{uv}O_v^T\}=\{\sPi_u\beta^{uv}\sPi_v^T\}$, where by signed permutations we mean matrices of the form $\sPi_u = \Pi_u R_u$, with $\Pi_u$ a permutation and $R_u$ a diagonal matrix with diagonal elements $\pm 1$. Therefore whenever considering the class of sets of $SO(3)$ transformations which preserve the diagonality of the matrices $\{\beta^{uv}\}$, it suffices to consider only the signed permutations. Of course this is trivially true when considering a single diagonal matrix, but is less obvious when considering {\em a set of matrices}.

The proof is achieved by introducing a consistent mapping from a given $SO(3)$ matrix $O$ to a signed permutation $\sPi(O)$, which we call the $\sPi$-reduction. This mapping will have the property that for a diagonal matrix $\beta$, if $O_L\beta O_R^T$ is diagonal, then $O_L\beta O_R^T=\sPi(O_L)\beta \sPi(O_R)^T$. Since the mapping is consistent, it also holds when considering the action of a set of $SO(3)$ rotations $\{O_u\}$ on a set of diagonal matrices $\{\beta^{uv}\}$.
 
Let us first define the $\sPi$-reduction. Consider the determinant of an $SO(3)$ rotation in a particular fixed basis $\{\hat{e}_m\}$:

$$\vert O \vert = \left\vert \left(\begin{array}{ccc}
a&b&c\\
d&e&f\\
g&h&i \end{array} \right)\right \vert= aei - afh - bdi + bfg +cdh -ceg $$

We can associate each of the monomials in the right hand side of the above equation with a matrix which is zero for all terms in the matrix which do not contribute to that monomial, and $\text{sign}(x)$ for all terms $x$ which do contribute to that monomial. So for example:

$$ afh \leftrightarrow \left(\begin{array}{ccc}
\text{sign}(a)&0&0\\
0&0&\text{sign}(f)\\
0&\text{sign}(h)&0 \end{array} \right)$$

We may also impose a lexicographical ordering on the matrices, so that the matrix associated with $bdi$ is lower in the ordering than the matrix associated with $bfg$. This is to ensure that the $D$-reduction has a consistent definition, and the particular choice of ordering is not important. We may then define the $\sPi$-reduction as follows:
\begin{definition}
Given a fixed basis $\{\hat{e}_m\}$, the $\sPi$-reduction $\sPi(O)$ of an $SO(3)$ operator $O$ is a signed permutation associated with the lowest monomial, by lexicographical ordering, in the determinant polynomial of $O$. For all entries in $O$ which do not contribute to the monomial, the corresponding entries in $\sPi(O)$ are $0$, and for all entries $x$ in $O$ which do contribute to the monomial, the corresponding entries in $\sPi(O)$ are $\text{sign}(x)$.
\end{definition}
Note that the determinant of $O$ is $+1$, so there always exists a non-zero monomial. Finally, note that the $\sPi$-reduction is always a signed permutation, but not always an $SO(3)$ matrix.

Let $\beta=\sum_m \lam{}{m} \hat{e}_{m} \hat{e}_{m}^T$ and 
$\tilde{\beta}=\sum_m \lamp{}{m} \hat{e}_{m} \hat{e}_{m}^T$.

\begin{lemma} \label{lamneqlamp} If $O_L \beta O_R^T = \tilde{\beta}$ and 
$\vert \lam{}{m} \vert \neq \vert \lamp{}{n} \vert \rightarrow [O_L]_{nm} = [O_R]_{nm}=0$.
\end{lemma}

\begin{proof}
$$\hat{e}_n^T \tilde{\beta}^2 O_R \hat{e}_m = \tilde{\lambda}_n^2 [O_R]_{nm}$$ $$\hat{e}_n^T O_R \beta^2  \hat{e}_m =  \lambda_m^2 [O_R]_{nm}$$ Since $ \tilde{\beta} = \tilde{\beta}^T$ and $\beta = \beta^T$ we have: $$\hat{e}_n^T \tilde{\beta}^2 O_R \hat{e}_m = \hat{e}_n^T O_R \beta^2  \hat{e}_m.$$ Therefore $$\tilde{\lambda}_n^2 [O_R]_{nm} = \lambda_m^2 [O_R]_{nm}.$$ A symmetric argument can be given to show that $$\tilde{\lambda}_n^2 [O_L]_{nm} = \lambda_m^2 [O_L]_{nm}.$$ Therefore if $\vert \lam{}{m} \vert \neq \vert \lamp{}{n}\vert$ then $[O_L]_{nm}=[O_R]_{nm}=0$.
\end{proof}

\begin{lemma} \label{opropo}If $O_L \beta O_R^T = \tilde{\beta}$ and $\lam{}{m} \neq 0 \rightarrow \forall n \;\;[O_L]_{nm} =\text{\rm sign}(\lam{}{m})\text{\rm sign}(\lamp{}{n} )[O_R]_{nm} $.
\end{lemma}

\begin{proof}
$$\hat{e}_n^T \tilde{\beta}O_R^T \hat{e}_m = \tilde{\lambda}_n [O_R]_{nm}$$
 $$\hat{e}_n^T O_L \tilde{\beta} \hat{e}_m = \lambda_m [O_L]_{nm}$$
  $$\hat{e}_n^T \tilde{\beta}O_R^T \hat{e}_m = \hat{e}_n^T O_L \tilde{\beta} \hat{e}_m$$ therefore 
  $$\tilde{\lambda}_n [O_R]_{nm}=\lambda_m [O_L]_{nm}.$$
 If $[O_R]_{nm}=0$ and $\lambda_m \neq 0$ then $[O_L]_{nm}=0$ and $$[O_L]_{nm} =\text{sign}(\lam{}{m})\text{sign}(\lamp{}{n} )[O_R]_{nm}.$$ Otherwise, if $[O_R]_{nm} \neq 0$ then by lemma \ref{lamneqlamp} $\vert \lam{}{m} \vert = \vert \lamp{}{n} \vert$ and if $\lambda_m \neq 0$ then $$[O_L]_{nm} =\text{sign}(\lam{}{m})\text{sign}(\lamp{}{n} )[O_R]_{nm}.$$
\end{proof}

\begin{lemma}\label{DproptoD}
Let $\sPi_L=\sPi(O_L)$ and $\sPi_R=\sPi(O_R)$. If $O_L \beta O_R^T =\tilde{\beta}$ and $\lam{}{m} \neq 0$ then $[\sPi_L]_{nm}  =\text{\rm sign}(\lam{}{m})\text{\rm sign}(\lamp{}{n}) [\sPi_R]_{nm} $.
\end{lemma}

\begin{proof}
If $\lam{}{m} \neq 0$ then by lemma \ref{opropo}  $[O_L]_{nm} =\text{sign}(\lam{}{m})\text{sign}(\lamp{}{n}) [O_R]_{nm}$. It follows that if $[\sPi_L]_{nm}\neq 0$ and $[\sPi_R]_{nm}\neq 0$ then $[\sPi_L]_{nm} =\text{sign}(\lam{}{m})\text{sign}(\lamp{}{n}) [\sPi_R]_{nm}$. Furthermore, if $[\sPi_L]_{nm}= 0$ and $[\sPi_R]_{nm}= 0$ then $[\sPi_L]_{nm} =\text{sign}(\lam{}{m})\text{sign}(\lamp{}{n}) [\sPi_R]_{nm}$. So it suffices to show that if $\lam{}{m} \neq 0$ then $[\sPi_L]_{nm}\neq 0$ if and only if $[\sPi_R]_{nm}\neq 0$.

Given a $\lam{}{m} \neq 0$, we may consider three cases, depending on the number of non-zero entries in the $m$th column of $O_L$. Note that every column of an $SO(3)$ matrix has at least one non-zero element, since otherwise it would have determinant 0.

\textbf{Case 1:} The $m$th column of $O_L$ has one non-zero entry. So there must exist a single $n \in \{1,2,3\}$ such that $[O_L]_{nm}\neq 0$ and, by lemma \ref{opropo}, $[O_R]_{nm} \neq 0$. Furthermore, if a column or row of $O_L$ has a single non-zero entry, then $\sPi(O_L)$ must have a non-zero entry at that same position. This can be seen by noting that every monomial in the determinant polynomial of $O_L$ must contain the non-zero entry $[O_L]_{nm}$, and so any matrix returned by $\sPi(O_u)$ must have a non-zero entry at that same position. Therefore $[\sPi_L]_{nm} \neq 0$ and $[\sPi_R]_{nm} \neq 0$.

\textbf{Case 2:} The $m$th column of $O_L$ has two non-zero entries. So there must exist $n_1, n_2 \in \{1,2,3\}$ such that $[O_L]_{n_1 m} \neq 0$, $[O_L]_{n_2 m}\neq 0$. By lemma \ref{opropo}, $[O_R]_{n_1 m} \neq 0$, $[O_R]_{n_2 m}\neq 0$. Furthermore, there exists $n_3\in \{1,2,3\}$ such that $[O_L]_{n_3 m} =0$ and, by lemma \ref{opropo}, $[O_R]_{n_3 m}= 0$. In other words $O_L$ and $O_R$ have the same block form, one $2\times 2$ block, and one $1 \times 1$ block. All elements in the blocks must be non-zero. This implies $O_L$ and $O_R$ have the same non-zero monomials in their determinant polynomial, and thus have the same smallest non-zero monomial. Therefore $[\sPi_L]_{nm} \neq 0$ if and only if $[\sPi_R]_{nm} \neq 0$.

\textbf{Case 3:} The $m$th column of $O_L$ has three non-zero entries. By lemma \ref{lamneqlamp} $\lam{}{n} \neq 0$ for all $n$. Therefore by lemma \ref{opropo} $[O_L]_{nm} = \text{sign}(\lam{}{m})\text{sign}(\lamp{}{n}) [O_R]_{nm}$ for all $n$ and $m$. This implies that $O_L$ and $O_R$ will always have the same non-zero monomials in their determinant polynomials. Therefore $[\sPi_L]_{nm} \neq 0$ if and only if $[\sPi_R]_{nm} \neq 0$.
\end{proof}

\begin{lemma} \label{Dreduction}
If $O_L \beta O_R^T =\tilde{\beta}$ then
$\sPi(O_L) \beta \sPi(O_R)^T= \tilde{\beta}$
\end{lemma}

\begin{proof}
Let $\sPi_x = \sPi(O_x)$.

$$\sPi_L \beta \sPi_R^T=\sum_{nn'}\sum_m \lam{}{m} [\sPi_L]_{nm} [\sPi_R]_{n'm} \hat{e}_{n}\hat{e}_{n'}^T$$

By lemma \ref{DproptoD}, if $\lam{}{m}\neq 0$ and $[\sPi_L]_{nm} \neq 0$ then $[\sPi_R]_{nm} \neq 0$. Since for any $m$ $[\sPi_R]_{nm}$ may only be nonzero for one $n$ it follows that if $\lam{}{m}\neq 0$ and $[\sPi_L]_{nm} \neq 0$ then $[\sPi_R]_{n'm}=[\sPi_R]_{nm} \delta_{nn'}=S(\lam{}{m})S(\lamp{}{n})[\sPi_L]_{nm} \delta_{nn'}$. Therefore

$$\sPi_L \beta \sPi_R^T=\sum_{n}\sum_m {\rm sign}(\lam{}{m}){\rm sign}(\lamp{}{n})\lam{}{m} ([\sPi_L]_{nm})^2  \hat{e}_n\hat{e}_{n}^T$$

$$\sPi_L \beta \sPi_R^T=\sum_{n}\sum_m {\rm sign}(\lamp{}{n})\vert \lam{}{m} \vert ([\sPi_L]_{nm})^2 \hat{e}_n\hat{e}_{n}^T$$

If $[\sPi_L]_{nm} \neq 0$ then $[O_L]_{nm}\neq 0$ 
and so $\vert \lam{}{m}\vert= \vert \lamp{}{n}\vert$. Therefore

$$[\sPi_L] \beta \sPi_R^T=\sum_{n}\sum_m {\rm sign}(\lamp{}{n})\vert \lamp{}{n}\vert ([\sPi_L]_{nm})^2 \hat{e}_n\hat{e}_{n}^T$$

$$\sPi_L \beta \sPi_R^T=\sum_{n} \lamp{}{n} \sum_m ([\sPi_L]_{nm})^2  \hat{e}_n\hat{e}_{n}^T$$

Since every column of $\sPi_L$ has only one element which is proportional to 1, $\sum_m ([\sPi_L]_{nm})^2=1$ for all $n$. Therefore

$$\sPi_L \beta \sPi_R^T=\sum_{n} \lamp{}{n} \hat{e}_n\hat{e}_{n}^T= O_L \beta O_R^T$$
\end{proof}

\section{Proof of Theorem \ref{thrm:NPcomplete}}\label{realnessIsNPComplete}

Clearly, the realness-under-Clifford rotations problem is in NP, since a prover can give the verifier the single-qubit Clifford rotations which make the terms in $H$ real. To prove the problem to be NP-hard, we show that if one can efficiently solve the realness-under-Cliffords problem, then one can also efficiently solve an NP-complete subclass of the exact cover problem, called restricted exact cover by three-sets.

 \begin{definition}
 \textbf{Restricted Exact Cover by Three-Sets (RXC3)}: Given a finite set of $3N$ elements $E=\{e_1, e_2... e_{3N} \}$, and a collection $R=\{S_1,S_2,\ldots\}$ of 3-element subsets of $E$, i.e. $|S_i|=3$, with the restriction that every element in $E$ appears in \textbf{exactly} three sets $S_i \in R$. Find a subcollection $R' \subset R$ such that every element in $E$ occurs in exactly one member of $R'$.
 \end{definition}
 
 It has been proven in \cite{Gonzalez1985} that RXC3 is NP-complete.
 
 Consider an instance of an RXC3 problem $A=(E, R)$. We can construct a corresponding Hamiltonian $H_A$ as follows. Let each element in $E$ correspond to a qubit. We can construct $H_A$ iteratively, starting with $H_A=0$. For every $e_i \in E$, define the pool of available Pauli operators $P_i =(X_i, Y_i, Z_i)$, 
where $X_i$ is a Pauli-$X$ on the $i$th qubit, and identity on all others. For every subset $(e_i, e_j, e_k) \in R$, take Paulis $\sigma_i$, $\sigma_j$, and $\sigma_k$ from the respective pools $P_i$, $P_j$, and $P_k$, reducing the pool sizes by one. 
Then let $H_A \rightarrow H_A + \sigma_i \sigma_j +\sigma_j \sigma_k + \sigma_i  \sigma_k$. Because of the restriction in RXC3 that every element $e_i \in E$ appears in exactly 3 subsets in $R$, we can be assured that our pools $P_i$ will always be able to supply a Pauli operator throughout this procedure, and will be exhausted by the end. Unless a particular ordering is imposed on $R$ and $P_i$, the definition of the above procedure is ambiguous. Any choice of ordering fixes the mapping. Let us suppose we have fixed a consistent ordering procedure for $R$ and every $P_i$.

For the sake of clarity, we include an example of an RX3C problem, and its corresponding Hamiltonian, with the pools ordered as $\{X,Y,Z\}$.

\noindent\textbf{RXC3 form}: 
\begin{equation}\begin{array}{cc} \label{RXC3Form}
E=&\{e_1,e_2,e_3,e_4,e_5,e_6\}\\ 
R=&(e_1,e_2,e_3),\\&(e_3,e_4,e_5),\\&(e_2,e_3,e_5),\\&(e_1,e_4,e_6),\\&(e_2,e_5,e_6),\\&(e_1,e_4,e_6)
\end{array}
\end{equation}
\noindent\textbf{Hamiltonian form}:
\begin{equation}\label{HamForm}\begin{array}{cccc}
H_A=&X_1 X_2 &+ X_2 X_3&+X_1 X_3\\
&+Y_3 X_4  &+X_4 X_5 &+Y_3 X_5\\
&+Y_2Z_3 &+Z_3 Y_5  &+Y_2  Y_5\\
&+Y_1  Y_4 &+ Y_4  X_6  &+ Y_1  X_6\\
&+Z_2 Z_5  &+Z_5 Y_6  &+Z_2  Y_6\\
&+Z_1 Z_4 &+ Z_4 Z_6 &+ Z_1 Z_6
\end{array}
\end{equation}

Given an RXC3 problem $A$, we can use the above procedure to produce a Hamiltonian $H_A$. We will now show that if we can efficiently determine whether or not $H_A$ can be made real under Clifford operations, then we can efficiently determine the answer to $A$.

Suppose $H_A$ can be made real under Clifford operations, then there exists some permutations on the qubits such that the only terms in the new Hamiltonian $H'_A$ which contain $Y$ operators are of the form $Y_i Y_j$. Recall that $H_A$ is composed exclusively of triplets of terms of the form $\sigma_i \sigma_j +\sigma_j \sigma_k + \sigma_i  \sigma_k$, with $\sigma_i$, $\sigma_j$, and $\sigma_k$ not appearing in any other terms. Therefore it must be the case that any of the terms in $H'_A$ of the form $Y_i Y_j$ can be grouped into triplets of the form $T_{ijk}=Y_i Y_j +Y_j Y_k + Y_i  Y_k$, with $Y_i$, $Y_j$, and $Y_k$ not appearing in any other terms. Every $T_{ijk}$ corresponds exactly with a $3$-element subset $(e_i, e_j, e_k) \in R$. Since $Y_i$, $Y_j$, and $Y_k$ do not appear in any other terms, it follows that the set of $3$-element subsets in $R$ associated with all such triplets $T_{ijk}$ form the exact covering $R'$. So if $H_A$ can be made real under single qubit Clifford rotations, then $A$ has an exact covering.

Suppose $A$ has an exact covering $R'$. For every $3$-element subset $(e_i, e_j, e_k) \in R'$ there exists a term $\sigma_i \sigma_j +\sigma_j \sigma_k + \sigma_i  \sigma_k$ in $H_A$, with $\sigma_i$, $\sigma_j$, and $\sigma_k$ not appearing in any other terms. Therefore we may choose to apply a sequence of permutations on the $i$, $j$ and $k$ qubits in $H_A$ so that $\sigma_i \rightarrow Y_i$, $\sigma_j \rightarrow Y_j$ and $\sigma_k \rightarrow Y_k$ with the guarantee that $Y_i$, $Y_j$, and $Y_k$ will then not appear in any other terms in our new Hamiltonian. We may continue in this fashion for every $3$-element subset in $R'$ without disturbing any of the previous qubits we have acted upon, and ultimately acting on every qubit. Therefore at the end of the procedure we will have produced a Hamiltonian $H'_A$ for which the only terms containing $Y$ operators are those of the form $Y_i Y_j$, and therefore is real. So if $A$ has an exact covering then $H_A$ can be made real by single qubit Clifford rotations.

\end{document}